\documentclass[12pt,twoside]{article}
\title{Bayesian detection of abnormal segments in multiple time series} 
\author{Lawrence Bardwell and Paul Fearnhead}
\date{\today}  

\usepackage[left = 2.5cm, right = 2.5cm, top = 2cm, bottom = 2.5cm]{geometry} 
\usepackage{setspace}
\usepackage{graphicx}
\usepackage{amsfonts} 
\usepackage{amsmath}
\usepackage{float}
\restylefloat{table}
\usepackage{natbib} 
\usepackage{placeins}
\usepackage{caption}
\usepackage{subcaption} 
\usepackage{amsthm}
\numberwithin{equation}{section}
\usepackage[ruled]{algorithm2e}
\usepackage{titlepic}
\DeclareMathAlphabet{\mathpzc}{OT1}{pzc}{m}{it}

\usepackage{color}
\usepackage[toc,page]{appendix}
\usepackage{fancyhdr} 
\usepackage{url}
\usepackage{xcolor}
\usepackage{bbm}
\usepackage{verbatim}
\newcommand{\emp}[1]{{\textcolor{blue}{#1}}}

\newtheorem{theorem}{Theorem}[section]

\newtheorem{lemma}[theorem]{Lemma}

\usepackage[pdftex]{hyperref}  
\graphicspath{{Figures/}}

\begin{document}
\bibliographystyle{apalike}  
\maketitle	
\begin{abstract}
We present a novel Bayesian approach to analysing multiple time-series with the aim of detecting abnormal regions. These are regions
where the properties of the data change from some normal or baseline behaviour. We allow for the possibility that
such changes will only be present in a, potentially small, subset of the time-series. We develop a general model for this problem, and
show how it is possible to accurately and efficiently perform Bayesian inference, based upon recursions that enable independent sampling
from the posterior distribution. 
A motivating application for this problem comes from detecting copy number variation (CNVs), using data from multiple
individuals. Pooling information across individuals can increase the power of detecting CNVs, but often a specific CNV
will only be present in a small subset of the individuals. We evaluate the Bayesian method on both simulated and real CNV data, and give
evidence that this approach is more accurate than a recently proposed method for analysing such data. 
\end{abstract}


%
\section{Introduction}
\label{sec:intro}


In this paper we consider the problem of detecting abnormal (or outlier) segments in multivariate time series. 
We assume that the series has some normal or baseline behaviour but that in certain intervals or segments of time a subset 
of the dimensions of the series has some kind of altered or abnormal behaviour. By the term abnormal behaviour
we mean some change in distribution of the data away from the baseline distribution. For example, this could include a change in mean,
variance, auto-correlation structure. In particular our work is 
concerned with situations where the size of this subset is only a small proportion of the total number of dimensions. 
We attempt to do this in a fully Bayesian framework.

This problem is increasingly common across a range of applications where the detection of 
abnormal segments (sometimes known as recurrent signal segments) is of interest (particularly 
in high dimensional and/or very noisy data). Some example applications include the analysis of 
the correlations between sensor data from different vehicles 
\citep{Spiegel:2011:PRC:2003653.2003657} or for intrusion detection in large interconnected 
computer networks \citep{1387011}. Another related application involves detecting common and 
potentially more subtle objects in a number of images, for example \citet{jin} and the 
references therein look at this in relation to multiple images taken of astronomical bodies.   

We will focus in particular on one specific example of this type of problem, namely that of detecting copy number variants (CNV's) in DNA sequences. 
A CNV is a type of structural variation that results in a genome having an abnormal (generally $\neq 2$) number of copies of a segment of DNA, such 
as a gene. Understanding these is important as these variants have been shown to account for much of the variability within 
a population. For a more detailed overview of this topic see \citet{overview_CNV,rare_copy_number} and the references therein. 

Data on CNVs for a given cell or individual is often in the form of ``log-R ratios'' for a range of probes, each associated with different locations
along the genome. These are calculated as log base 2 of the ratio of the measured probe intensity to the reference intensity for a given probe. Normal
regions of the genome would have log-R ratios with a mean of 0, whereas CNVs would have log-R ratios with a mean that is away from zero. 

\begin{figure}[t!]
\centering
\includegraphics[scale=0.75]{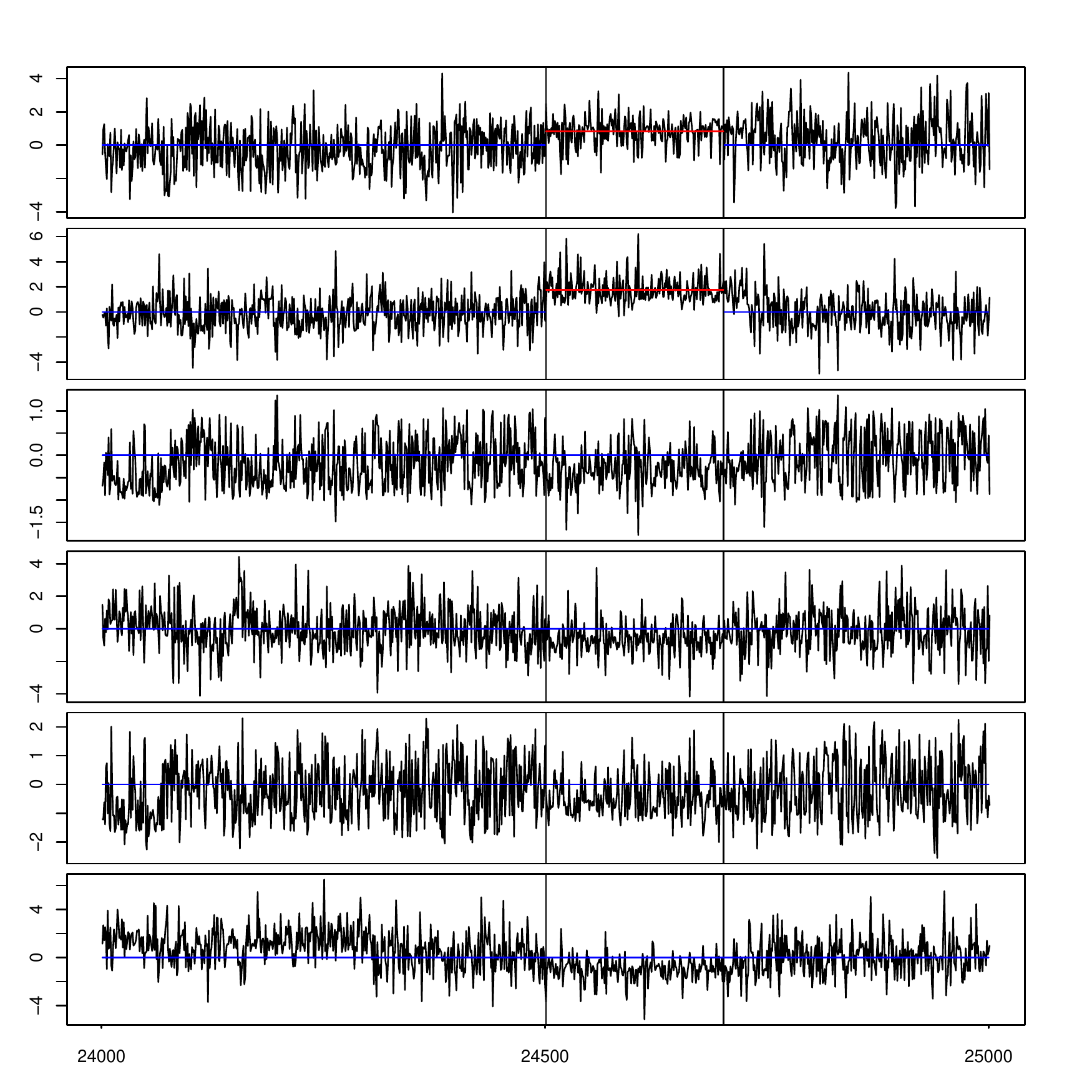}
\caption{Log-R ratios from 6 individuals for a small portion of chromosome 16. We indicate the baseline level (mean zero) by a horizontal line in blue and the identified CNV (abnormal region) is highlighted between two vertical black lines with the mean of the affected individuals in red. For this CNV only the first two individuals (NA10851 and NA12239) are affected.}
\label{fig:real_data_examp}
\end{figure}

Figure \ref{fig:real_data_examp} gives an example of such data from 6 individuals. We can see that there is substantial noise in the data, and
each CNV may cover only a relatively small region of the genome. Both these factors mean that it can be difficult to accurately detect CNVs by
analysing data from a single individual or cell. To increase the power to identify CNVs we can pool information by jointly analysing data 
from multiple individuals. However this is complicated as a CNV may be observed for only a subset of the individuals. For example, for the data in 
Figure \ref{fig:real_data_examp}, which shows data from a small portion of  chromosome 16,
we have identified a single CNV which affects only the first two individuals. 
This can seen by the raised means (indicated by the red lines) in these 
two series for a segment of data. By comparison, the other individuals 
are unaffected in this segment. 


Whilst there has been substantial research into methods for detecting outliers \cite[]{Tsay01122000,2006} or abrupt changes in data 
\cite[]{Olshen01102004,JTSA:JTSA12035,WYSE2011,Frick:2014}, the problem
of identifying outlier regions in just a subset of dimensions has received less attention. Exceptions include methods 
 described in \citet{detec_simul_gene_copy} and \citet{siegmund2011}. However \cite{rare_copy_number} argue that these
 methods are only able to detect common variants, that is abnormal segments for which a large proportion of the dimensions have undergone the 
 change.  \cite{rare_copy_number} propose a method, the PASS algorithm, which is also able to detect rare variants. 
 
The methods of \cite{siegmund2011} and \cite{rare_copy_number} are based on defining an appropriate test-statistic for whether a region is abnormal
for a subset of dimensions, and then recursively using this test-statistic to identify abnormal regions. As such the output of these methods is a 
point estimate of the which are the abnormal regions. Here we introduce a Bayesian approach to detecting abnormal regions. This is able to both
give point estimates of the number and location of the abnormal regions, and also to give measures of uncertainty about these. We show how it is
possible to efficiently simulate from the posterior distribution of the number and location of abnormal regions, through
using recursions similar to those from multiple changepoint detection \cite[]{Barry/Hartigan:1992,Fearnhead2006,FearnheadVasileiou2009}. We call the resulting algorithm,
Bayesian Abnormal Region Detector (BARD). 

The outline of the paper is as follows. In the next section we introduce our model, both for the general problem of detecting abnormal regions, and 
also for the specific CNV application. In Section \ref{sec:inference} we derive the recursions that enable us to draw iid samples from the posterior,
as well as a simple approximation to these recursions that results in an algorithm, BARD, that scales linearly with the length of data set. We then present theoretical results that
show that BARD can consistently estimate the absence of abnormal segments, and the location of any abnormal segments, and is robust to some mis-specification of the priors.
In Section \ref{sec:results} we evaluate BARD for the CNV application on both simulated and real data. Our results suggest that BARD is more accurate than PASS,
particularly in terms of having fewer fales positives. Furthermore, we see evidence that posterior probabilities are well-calibrated and hence
are accurately representing the uncertainty in the inferences. The paper ends with a discussion.

\section{The Model}
\label{sec:model}

We shall now describe the details of our model. Consider a multiple time series of dimension $d$ and length $n$,
$\mathbf{Y}_{1:n} = ( \mathbf{Y}_{1}, \mathbf{Y}_{2} , \hdots , \mathbf{Y}_{n} )$  where  $\mathbf{Y}_{i} = \left( Y_{i,1} , Y_{i,2} , \hdots , Y_{i,d} \right)^{T}.$ We model this
data through introducing a hidden state process, $X_{1:n}$. The hidden state process will contain information about where the abnormal segments of the data are. Our model is defined through specifying the
distribution of the hidden state process, $p(x_{1:n})$, and the conditional distribution of the data given the state process, $p(\mathbf{y}_{1:n} | x_{1:n})$.  These are defined in Sections \ref{sec:state_model} and
\ref{sec:like_model} respectively.

Our interest lies in inference about this hidden state process given the observations. This involves calculating the posterior distribution for the states
\begin{align} \label{eq:hidden_state_post}
p(x_{1:n} | \mathbf{y}_{1:n}) \propto p(x_{1:n} , \mathbf{y}_{1:n}) = p(x_{1:n}) p(\mathbf{y}_{1:n} | x_{1:n}).
\end{align}
It should be noted that these probabilities will depend on a set of hyper-parameters. These parameters are initially assumed to be known, however we will later discuss performing inference for them.    

\subsection{Hidden State Model}
\label{sec:state_model}

The hidden state process will define the location of the abnormal segments. We will model the location of these segments through a renewal process. The length of a given segment is drawn 
from some distribution which depends on the segment type, and is independent of all other segment lengths. We assume a normal segment is always followed by an abnormal segment, but allow for either a normal or abnormal
segment to follow an abnormal one. The latter is because each abnormal segment may be abnormal in a different way, for example with different subsets of the time-series being affected. 
This will become clearer when we discuss the likelihood model in Section~\ref{sec:like_model}.   

To define such a model we need distributions for the lengths of normal and abnormal segments. We denote the cumulative distribution functions of these lengths by $G_N(t)$ and $G_A(t)$ respectively. We also need to specify the probability that an abnormal segment is followed by either a normal or abnormal segment. We denote these
probabilities as $\pi_N$ and $\pi_A$ respectively, with $\pi_N=1-\pi_A$. 

Note that the first segment for the data will have a different distribution to other segments as it may have started at some time prior to when we started collecting data. We can define this distribution in a way that
is consistent with our underlying model by assuming the process for the segments is at stationarity and that we start observing it at an arbitrary time. Renewal theory \citep{cox1962renewal} then gives the distribution function for the length of the first segment. If the first segment is normal, then we define its cumulative distribution function as
\[
 G_{0N}(t)=\sum_{s=1}^t \frac{1-G_N(s)}{E_N},
\]
where $E_N$ is the expected length of a normal segment. The cumulative distribution function for the first segment conditional on it being abnormal, $G_{0A}(t)$, is similarly defined.

Formally, we define our hidden state process $X_t$ as $X_t = (C_t,B_t)$ where $C_t$ is the end of the previous segment prior to time $t$ and $B_t$ is the type of the current segment. So $C_t\in\{0,\ldots,t-1\}$ with $C_t=0$ 
denoting that the current segment is the first segment. We use the notation that $B_t=N$ if the 
current segment is normal, and $B_t=A$ if not. This state process is Markov, and thus we can write 
 \begin{align} \label{eq:hidden_state_decomp}
\begin{split}
p(x_{1:n}) &= p(c_{1:n},b_{1:n}) \\
&= \Pr(C_1=c_1,B_1=b_1) \prod_{i=1}^{n-1} \Pr( C_{i+1}=c_{i+1} , B_{i+1}=b_{i+1} |C_{i} = c_{i} ,B_{i}= b_{i} ).
\end{split}
\end{align}
The decomposition in \eqref{eq:hidden_state_decomp} gives us two aspects of the process to define, namely the transition probabilities $\Pr( C_{i+1}=c_{i+1} , B_{i+1}=b_{i+1} | c_{i} , b_{i} )$ 
and the initial distribution, $\Pr(C_1=c_1,B_1=b_1)$.

Firstly consider the transition probabilities. Now either $C_{t+1}=C_t$ or $C_{t+1}=t$ depending on whether a new segment starts between time $t$ and $t+1$. The probability of a new segment starting
is just the conditional probability of a segment being of length $t-C_t$ given that is at least $t-C_t$. If $C_{t+1}=C_t$, then we must have $B_{t+1}=B_t$, otherwise the distribution of the type of the new segment
depends on the type of the previous segment as described above.

Thus for $i=1,\ldots,t-1$ we have
\begin{align}
   \begin{split}
\Pr(C_{t+1}=j , B_{t+1} = k | C_t = i , B_t = N) &= \begin{cases} 
     \frac{1-G_{N}(t-i) }{1-G_{N}(t-i-1)} & \textrm{ if $j=i$ and  $k=N$,} \\
     \frac{ G_{N}(t-i) - G_{N}(t-i-1) }{ 1 - G_{N}(t-i-1) } & \textrm{ if $j=t$ and  $k=A$,} \\
	0 & \textrm{ otherwise,}
   \end{cases}  \\
\Pr(C_{t+1}=j , B_{t+1} = k | C_t = i , B_t = A) &= \begin{cases} 
      \frac{ 1 - G_{A}(t-i) }{1-G_{A}(t-i-1)} & \textrm{ if $j=i$ and  $k=A$,} \\
     \pi_{A} \left( \frac{ G_{A}(t-i) - G_{A}(t-i-1) }{ 1 - G_{A}(t-i-1) } \right) & \textrm{ if $j=t$ and  $k=A$,} \\
	  \pi_{N}  \left(  \frac{ G_{A}(t-i) - G_{A}(t-i-1) }{ 1 - G_{A}(t-i-1) }  \right)  & \textrm{ if $j=t$ and  $k=N$,} \\
     0 & \textrm{ otherwise}.
   \end{cases}
   \end{split}
\end{align}
For $i=0$, that is when $C_t=0$, we replace $G_N(\cdot)$ and $G_A{\cdot}$ with $G_{0N}(\cdot)$ and $G_{0A}(\cdot)$ respectively.

Finally we need to define the initial distribution for $X_1=(B_1,C_1)$. Firstly note that $C_1=0$ so we need only the distribution of $B_1$. We define this as the stationary distribution of the $B_t$ process.
This is \cite[see for example Theorem 5.6 of][]{kulkarni2012introduction}
\[
 \Pr(B_1=N)=\frac{ \pi_{N}E_{N} }{ \pi_{N}E_{N} + E_{A} },~~~\Pr(B_1=A)=1-\Pr(B_1=N),
\]
where $E_N$ and $E_A$ are the expected lengths of normal and abnormal segments respectively.
 
\subsection{Likelihood model}
\label{sec:like_model}

The hidden process $X_{1:n}$ described above partitions the time interval into contiguous non-overlapping segments each of which is either normal, $N$, or abnormal, $A$. 
Now conditional on this process we want to define a likelihood for the observations, $p(\mathbf{y}_{1:n} | x_{1:n} )$.

To make this model tractable we assume a conditional independence property between segments, this means that if we knew the locations of segments and their types then data from different segments are independent. 
Thus when we condition on $C_t$ and $B_t$ the likelihood for the first $t$ observations factorises as follows     
\begin{align} \label{eq:cond_ind}
p(\mathbf{y}_{1:t} | C_t =j , B_t) = p( \mathbf{y}_{1:j} | C_t=j,B_t) p(\mathbf{y}_{j+1:t} |C_t=j, B_t). 
\end{align}

The second term in equation \eqref{eq:cond_ind} is the marginal likelihood of the data, $\mathbf{Y}_{j+1:t}$, given it comes from a segment that has type $B_t$. 
We introduce the following notation for these segment marginal likelihoods, where for $s\geq t$,
\begin{align} \label{eq:marg_lik}
\begin{split}
&P_N(t,s) = \Pr(\mathbf{y}_{t:s} |C_s=t-1, B_s = N  ), \\
&P_A(t,s) = \Pr(\mathbf{y}_{t:s} |C_s =t-1, B_s = A ),
\end{split}
\end{align}
and define $P_N(t,s)=1$ and $P_A(t,s)=1$ if $s<t$.

Now using the above factorisation we can write down the likelihood conditional on the hidden process. Note that we can condition on $X_t$ rather than the full history $X_{1:n}$ in each of the factors in \eqref{eq:likelihood_for_model} due to the conditional independence assumption on the segments 
\begin{align} \label{eq:likelihood_for_model}
\begin{split}
  p(\mathbf{y}_{1:n} | x_{1:n} ) &= \prod_{t=1}^{n} p( \mathbf{y}_t | x_{1:n} , \mathbf{y}_{1:(t-1)} ) \\
&= \prod_{t=1}^{n} p( \mathbf{y}_t | C_t , B_t , \mathbf{y}_{(C_t + 1):(t-1)} ). 
\end{split}
\end{align}
The terms on the right-hand side of equation~\eqref{eq:likelihood_for_model} can then be written  in terms of the segment marginal likelihoods 
\begin{align} \label{eq:cond_like}
p( \mathbf{y}_t | C_t , B_t , \mathbf{y}_{(C_t + 1):(t-1)} )
= \frac{ P_{B_t}(C_t + 1 , t) }{ P_{B_t}(C_t + 1 , t-1) }.
\end{align}
Thus our likelihood is specified through defining appropriate forms for the marginal likelihoods for normal and abnormal segments.

\subsubsection{Model for data in normal segments}

For a normal segment we model that the data for all dimensions of the series are realisations from some known distribution, $\mathcal{D}$, and these realisations are independent over both time and dimension. 
Denote the density function of the distribution $\mathcal{D}$ as $f_{\mathcal{D}}(\cdot)$. We can write down the segment marginal likelihood as
\begin{align}
\label{eq:normal_marg}
P_N(t,s) = \prod_{k=1}^{d} \prod_{i=t}^{s} f_{\mathcal{D}}(y_{i,k}). 
\end{align}

\subsubsection{Model for data in abnormal segments}

For abnormal segments our model is that data for a subset of the dimensions are drawn from $\mathcal{D}$, with the data for the remaining dimensions being independent realisations from a different distribution, $\mathcal{P}_{\theta}$, which depends on a segment specific parameter $\theta$. We denote the density function for this distribution as $f_{\mathcal{P}}(\cdot| \theta)$.

Our model for which dimensions have data drawn from $\mathcal{P}_{\theta}$ is that this occurs for dimension $k$ with probability $p_k$, independently of the other dimensions. Thus if we have an abnormal
segment with data $\mathcal{Y}_{t:s}$, with segment parameter $\theta$, the likelihood of the data associated with the $k$th dimension is
\[
 p_k\prod_{i=t}^s f_{\mathcal{P}}(y_{i,k}|\theta)+(1-p_k)\prod_{i=t}^s f_{\mathcal{D}}(y_{i,k}).
\]
Thus by independence over dimension
\[
 p(\mathbf{y}_{t:s}|\theta)=\prod_{k=1}^d \left( p_k\prod_{i=t}^s f_{\mathcal{P}}(y_{i,k}|\theta)+(1-p_k)\prod_{i=t}^s f_{\mathcal{D}}(y_{i,k}) \right).
\]

Our model is completed by a prior for $\theta$, $\pi(\theta)$. 
To find the marginal likelihood $P_A(t,s)$ we need to integrate out $\theta$ from $p(\mathbf{y}_{t:s}|\theta)$
\begin{align}
\label{eq:abnormal_marg}
P_A(t,s) = \int p( \mathbf{y}_{t:s} | \theta  ) \pi(\theta) \,\mbox{d} \theta.
\end{align}
In practice this integral will need to be calculated numerically, which is feasible if $\theta$ is low-dimensional.

\subsubsection{CNV example}
\label{sec:cnv_example}


In Section \ref{sec:intro} we discussed the copy number variant (CNV) application and showed some real data in Figure \ref{fig:real_data_examp}.
From the framework described above we now need to specify a model for normal and abnormal segments. 
Following \citet{rare_copy_number} we model the data as being normally distributed with constant variance but differing means either zero or $\mu$ depending on whether we are in a normal or abnormal segment. This model also underpins the simulation studies that we present in Section \ref{sec:results}. 

Using the notation from the more general framework discussed above the two distributions for normal and abnormal segments are  
\begin{align*}
\mathcal{D} &\sim N(0,\sigma^2)  \\
\mathcal{P}_{\mu} &\sim N(\mu,\sigma^2).
\end{align*} 
We assume that the variance $\sigma^2$ is constant and known (in practice we would be able to estimate it from the data). 

Having specified these two distributions we then need to calculate marginal likelihoods for normal and abnormal segments given by equations 
\eqref{eq:normal_marg} and \eqref{eq:abnormal_marg} respectively. Calculating the marginal likelihood for a normal segment is simple because 
of independence over time and dimension as shown in equation \eqref{eq:normal_marg}. However calculating $P_A(\cdot , \cdot)$ is more challenging,
as there is no conjugacy between $p(\mathbf{y}|\mu)$ and $\pi(\mu)$ so we can only numerically approximate the integral. Calculating the
numerical approximation is fast as it is a one-dimensional integral.

In the simulation studies and results we take the prior for $\mu$ to be uniform on a region that excludes values of $\mu$ close to zero. For CNV data such a prior seems reasonable empirically (see Figure \ref{fig:ab_mu}) and also because we expect CNV's to correspond to a change in mean level of at least $\log(3/2)$ and can be both positive or negative.  


\section{Inference} \label{sec:inference}

We now consider performing inference for the model described in Section \ref{sec:model}. Firstly a set of recursions to perform this task exactly are introduced and 
then an approximation is considered to make this procedure computationally more efficient.

\subsection{Exact On-line inference}
\label{sec:recursions}

We follow the method of \citet{FearnheadVasileiou2009} in developing a set of recursions for the posterior distribution of the hidden state, the location of the start of the current segment 
and its type, at time $t$ given that we have observed data upto time $t$,
$p( x_t | \mathbf{y}_{1:t}) = p(c_t , b_t | \mathbf{y}_{1:t} )$,
for $t \in \{ 1, 2 , \hdots , n \}$.
These  are known as the filtering distributions. Eventually we will be able to use these to simulate from the full posterior,
$p( x_{1:n} | \mathbf{y}_{1:n} )$.

To find these filtering distribution we develop a set of recursions that enable us to calculate 
$ p(c_{t+1} , b_{t+1} | \mathbf{y}_{(1:t+1)} )$ in terms of $p(c_t , b_t | \mathbf{y}_{1:t} )$. 
These recursions are analogous to the forward-backward equations widely used in analysing Hidden Markov models. 

There are two forms of these recursions 
depending on whether   $C_{t+1}=j$ for $j < t$ or $C_{t+1}=t$. We derive the two forms separately. Consider the first case.  For $j<t$ and $k\in\{N,A\}$,
\begin{eqnarray*}
\lefteqn{ p(C_{t+1} = j , B_{t+1} = k | \mathbf{y}_{1:(t+1)}) \propto p(\mathbf{y}_{t+1} | \mathbf{y}_{1:t} , C_{t+1}=j , B_{t+1} = k)p(C_{t+1} = j, B_{t+1} = k | \mathbf{y}_{1:t}  )}& &\\
&=&\left(\frac{  P_{k}(j+1,t+1)  }{  P_{k}(j+1,t)   }\right)\Pr(C_{t+1}=j , B_{t+1} = k | C_t = j , B_t = k)p(C_t = j , B_t = k | \mathbf{y}_{1:t}),
\end{eqnarray*}
where the first term in the last expression is the conditional likelihood from equation \eqref{eq:cond_like}. The second two terms use the fact that there has not been a new segment and hence $C_{t+1}=C_t$ and $B_{t+1}=B_t$.

Now for the second case, when $C_{t+1} = t$,
\begin{align*}
p(C_{t+1}&=t,B_{t+1}=k|\mathbf{y}_{1:t})  \\
&=\sum_{i=0}^{t-1} \sum_{l \in \{N,A\}} p(C_t=i,B_t=l|\mathbf{y}_{1:t})\Pr(C_{t+1} = t , B_{t+1} = k | C_t = i , B_t = l ).  
\end{align*}
Thus, as $p(\mathbf{y}_{t+1} | C_{t+1} = t, B_{t+1} = k, \mathbf{y}_{1:t}) = P_{k}(t+1,t+1)$, 
the filtering recursion is;
\begin{align*}
p( C_{t+1}&=t , B_{t+1}=k | \mathbf{y}_{1:(t+1)}) \propto \\
&P_{k}(t+1,t+1) \sum_{i=0}^{t-1} \sum_{l \in \{N,A\}} p(C_t=i, B_t=l|\mathbf{y}_{1:t}) \Pr(C_{t+1} = t , B_{t+1} = k | C_t = i , B_t = l ).
\end{align*}
These recursions are initialised by $p( C_1 = 0 , B_1 = k | \mathbf{y}_1 ) \propto \Pr(B_1=k) P_{k}(1,1)$ for $k \in \{ N , A \}$. 

\subsection{Approximate Inference} \label{sec:particle_approx}
The support of the filtering distribution $p(c_{t},b_{t}|\mathbf{y}_{1:t})$ has $2t$ points. Hence, calculating $p(c_{t},b_{t}|\mathbf{y}_{1:t})$ exactly is of order $t$ both in terms of computational and storage costs. 
The cost of calculating and storing the full set of filtering distributions $t=1,2,\hdots,n$ is thus of order $n^2$. For larger data sets this exact calculation 
can be prohibitive. 
A natural way to make this more efficient is to approximate each of the filtering distributions by distributions with a fewer number of support points. 
In practice such an approximation is feasible as many of the support points of each filtering distribution have negligible probability. If we removed these points then we could greatly increase the speed of our algorithm without sacrificing too much accuracy. 

We use the stratified rejection control (SRC) algorithm \citep{FearnheadLiu2007} to produce an approximation to the filtering distribution with potentially fewer support points at each time-point. 
This algorithm requires the choice of a threshold, $\alpha\geq 0$. At each iteration the SRC algorithm keeps all support points which have a probability greater than $\alpha$. 
For the remaining particles the probability of them being removed is proportional to their associated probability and the resampling is done in a stratified manner. 
This algorithm has good theoretical properties in terms of the error introduced at each resampling step, measured by the Kolmogorov Smirnov distance, being bounded by $\alpha$. 

\subsection{Simulation}
\label{sec:backwards_sim}
Having calculated and stored the filtering distributions, either exactly or approximately, simulating from the posterior is straightforward. This is performed by simulating the hidden process backwards in time. 
First we simulate $X_n=(C_n,B_n)$ from the final filtering distribution $p(c_n,b_n|\mathbf{y}_{1:n})$. Assume we simulate $C_n=t$. Then, by definition of the hidden process, we have $C_s=t$ and $B_s=B_n$ for $s=t+1,\ldots,n-1$, as
these time-points are all part of the same segment. Thus we next need to simulate $C_t$, from its conditional distribution given $C_{t+1}$, $B_{t+1}$ and $\mathbf{Y}_{1:n}$,
\begin{align*} 
   \begin{split}
p(c_t , b_t &| C_{t+1} = t , B_{t+1} , \mathbf{y}_{1:n} ) \\
& \propto p(c_t , b_t , C_{t+1} = t ,B_{t+1} , \mathbf{y}_{1:n}) \\
& = p(c_t,b_t)\Pr(C_{t+1}=t,B_{t+1}|C_t,B_t)p(\mathbf{y}_{1:n} | C_t,B_t,C_{t+1}=t,B_{t+1}) \\
& \propto p(c_t,b_t)\Pr(C_{t+1}=t,B_{t+1}|C_t,B_t)p(\mathbf{y}_{1:t} | C_t,B_t) \\
& \propto p(c_t,b_t| \mathbf{y}_{1:t})\Pr(C_{t+1}=t,B_{t+1}|C_t,B_t).
   \end{split}
\end{align*}
We then repeat this process, going backwards in time until we simulate $C_t=0$. From the simulated values we can extract the location and type of each segment.

\subsection{Hyper-parameters}
\label{sec:hyp_param}

As mentioned earlier in Section \ref{sec:model} the posterior of interest \eqref{eq:hidden_state_post} depends upon a vector of hyper-parameters which we now label as $\Psi$. 
Specifically $\Psi$ contains the parameters for the LOS distributions for the two differing types of segments which determine the cdf's $G_{N}(\cdot)$ and $G_{A}(\cdot)$. 

We use two approaches to estimating these hyper-parameters. The first is to maximise the marginal-likelihood for the hyper-parameters, which we can do using Monte Carlo EM (MCEM). 
For general details on MCEM see \citet{Levine2001}.
Although convergence of the hyper-parameters is quite rapid in the examples we look at in Section \ref{sec:results}, for very large data sets a cruder but faster alternative is to initially segment the data using a 
different method to ours and then use information from this segmentation to inform the choice of hyper-parameter values. 
The alternative method we use is the PASS method of \citet{rare_copy_number} and discussed in detail in Section \ref{sec:results}.    

\subsection{Estimating a Segmentation}
\label{sec:seg_est}
We have described how to calculate the posterior density $p(x_{1:n} | \mathbf{y}_{1:n})$ from which we can easily draw a large number of samples. However we often want to report a single estimated ``best'' segmentation
of the data.
We can define such a segmentation using
Bayesian decision theory \citep{decision_theory_berger}. 
This involves defining a loss function which determines the cost of us making a mistake in our estimate of the true quantity which we then seek to minimise. There are various choices of loss function
we could use \cite[see][]{loss_holmes}, but we use a loss that is a sum of a loss for estimating whether each location is abnormal or not.
If $L(\tilde{b}_t | b_t)$ gives the cost of making the decision that the state at time $t$ is $\tilde{b}_t$ when in fact it is $b_t$, then: 
\begin{align} \label{eq:assym_loss}
L(\tilde{b}_t | b_t)= \begin{cases} 
     1 & \textrm{ if $\tilde{b}_t=$ A and $b_t=$ N} \\
     \gamma & \textrm{ if $\tilde{b}_t=$ N and  $b_t=$ A} \\
		0 & \textrm{ otherwise} \\
   \end{cases} 
\end{align}
The inclusion of $\gamma$ allows us to vary the relative penalty for false positives as compared to false negatives. Under this loss
 we estimate $\hat{b}_t=N$ if $\pi(b_t = A) < 1/(1+\gamma) $ or $\hat{b}_t=A$ otherwise.


\section{Asymptotic Consistency} \label{asymptotics}

We will now consider the asymptotic properties of the method as $d$, the number of time-series, increases. Our aim is to study the robustness of inferences to the choice
of prior for the abnormal segments, and the estimate of $p_d$, allowing for abnormal segments that are rare.  We will assume that each time-series is of fixed length $n$. Following \cite{rare_copy_number}, 
to consider the influence of rare abnormal segments, we will let the proportion
of sequences that are abnormal in an abnormal segment to decrease as $d$ increases.

Our assumptions on how the data is generated is that there are a fixed number and location of abnormal segments. We will assume the model  of Section \ref{sec:cnv_example} with, without loss of generality, 
$\sigma^2=1$. So if $B_t=N$, then $Y_{i,j} \sim N(0,1)$. If
$(t,\ldots,s)$ is an abnormal segment then it has an associated mean, $\mu_0\neq 0$. For each $j=1,\ldots,d$, independently with probability $\alpha_d$, $Y_{i,j} \sim N(\mu,1)$ for $i=t,\ldots,s$; 
otherwise $Y_{i,j} \sim N(0,1)$ for $i=t,\ldots,s$. 

We fit the model of Section \ref{sec:model}, assuming the correct likelihood for data in normal and abnormal segments. For each abnormal segment we will have an independent prior for the associated mean, $\pi(\mu)$. 
Our assumptions on $\pi(\mu)$ is that its support is a subset of $\{[-b,-a],[a,b]\}$ for some $a>0$ and $b<\infty$, and it places non-zero probability on both positive and negative values of $\mu$. The model we
fit will assume a specified probability, $p_d$, of each sequences being abnormal within each abnormal segment. Note that we do not require $p_d=\alpha_d$, the true probability, but we do allow the choice of this
parameter to depend on $d$. 

The Lemmas used in the proof of the following two theorems can be found in the appendices.


\begin{theorem}
\label{thm:1}
Assume the model for the data and the constraints on the prior specified above. Let $\mathcal{E}$ be the event that there are no abnormal segments, and $\mathcal{E}^{c}$ its complement. 
If there are no abnormal segments and  $d\rightarrow \infty$, with $1/p_d = O(d^{\frac{1}{2}-\epsilon})$ 
for some $\epsilon>0$, then
\[
 \Pr(\mathcal{E}^{c}|\mathbf{y}_{1:n})\rightarrow 0,
\]
in probability.
\end{theorem}
\begin{proof}
As $n$ is fixed, we have a fixed number of possible segmentations. We will show that the posterior probabilility of each possible segmentation with at least one abnormal segment is $o_p(1)$ as $d\rightarrow \infty$. 

For time-series $k$ let $P_{N,k}(t,s)$ denote the likelihood of the data $y_{t,k},\ldots,y_{s,k}$ assuming this is a normal segment; 
and let $P_{A,k}(t,s;\mu)$ be the marginal likelihood of the same data given that it is drawn
from independent Gaussian distributions with mean $\mu$. Then if we have a segmentation with $m$ abnormal segments, with the $i$th abnormal segment from $t_i$ to $s_i$,
the ratio of the posterior probability of this segmentation to the posterior
probability of $\mathcal{E}$ is
\[
K \prod_{i=1}^m \left( \int \left\{ \prod_{k=1}^d \frac{P_{A,k}(t_m,s_m;\mu)}{P_{N,k}(t_m,s_m)}
\right\} \pi(\mu) \mbox{d}\mu \right),
\]
where $K$ is the ratio of the prior probabilities of these two segmentations. So it is sufficient to show that for all $t\leq s$, 
\begin{align}
\label{eq:first_int}
  \int \left\{ \prod_{k=1}^d \frac{P_{A,k}(t,s;\mu)}{P_{N,k}(t,s)}
\right\} \pi(\mu) \mbox{d}\mu \rightarrow 0
\end{align}
in probability as $d\rightarrow \infty$.
 
Our limit involves treating the data as random. Each term in this product is then random, and of the form 
\begin{equation} \label{eq:P1}
 \frac{P_{A,k}(t,s;\mu)}{P_{N,k}(t,s)}=1+p_d \left( \exp\left\{\mu\sum_{u=t}^s \left( Y_{k,u}- \frac{\mu}{2} \right) \right\}-1 \right).
\end{equation} 
By applying Lemma \ref{case_1} separately to positive and negative values of $\mu$, 
we have that this tends to 0 with probability 1 as $d\rightarrow \infty$. This is true for all possible segmentations with at least one abnormal segments. As $n$ is fixed, there are a finite 
number of such segments, so the result follows.

\end{proof}

Theorem \ref{thm:2} tells us that the posterior probability of misclassifying a time point as normal when it is abnormal tends to zero as more time-series are observed. 
\begin{theorem}
\label{thm:2}
Assume the model for the data and the constraints on the prior specified above.
 Fix any position $t$, and consider the limit as $d\rightarrow \infty$, with $d p_d^2\rightarrow \infty$ and either 
 \begin{itemize}
  \item[(i)] $p_d = o(\alpha_d)$; or
  \item[(ii)] if $\mu_0$ is the mean associated with the abnormal sequences at position $t$, then there exists a region $A$ such that the prior probability associated with $\mu \in A$ is non-zero, and for all $\mu\in A$ 
  and for sufficiently large $d$ 
  \[
   \alpha_d \left( e^{\mu\mu_0} - 1 \right) - \frac{p_d}{2} \left( e^{\mu^2} - 1 \right) > 0.
  \]
 \end{itemize}
 Then if $B_t=A$ 
  \[
  \Pr(B_t=N|\mathbf{y}_{1:n}) \rightarrow 0. 
 \]
 in probability.
\end{theorem}
\begin{proof}
We will show that each segmentation with $B_t=N$ has posterior probability that tends to 0 in probability as $d\rightarrow \infty$. For each segmentation with $B_t=N$ we will compare its posterior probability 
with one which is identical except for the addition of an abnormal segmentant, of length 1, at location $t$. The ratio of posterior probabilities of these two segmentations will be
\[
 K  \left( \int \left\{ \prod_{k=1}^d \frac{P_{A,k}(t,t;\mu)}{P_{N,k}(t,t)}
\right\} \pi(\mu) \mbox{d}\mu \right),
\]
where $K$ is a constant that depends on the prior for the segmentations. We require that this ratio tends to infinity in probability as $d\rightarrow \infty$. Under both conditions (i) and (ii) above this follows immediately
from Lemma \ref{case_2}. For case (i) we are using the fact that the prior places positive probability both on $\mu$ being positive and negative, and for $\mu$ the same sign as $\mu_0$ we have that 
$e^{\mu\mu_0}>1$.
\end{proof}

This result shows some robustness of the Bayesian approach to the choice of prior. 
Consider a prior on the mean for an abnormal segment that has strictly positive density for values in $\{[-b,-a],[a,b]\}$ for $a>0$. 
Then for any true mean, $\mu_0$ with 
$|\mu_0|\geq a$, we will consistently estimate the segment as abnormal provided the assumed or estimated
probability of a sequence being abnormal is less than twice the true value. 
Thus we want to choose $a$ to be the smallest absolute value of the mean of an abnormal segment we expect or wish to detect. The choice of $b$ is less important, in that it does not affect the asymptotic consistency
implied by the above theorem. 

Furthermore we do not need to specify $p_d$ exactly for consistency -- the key is not to over-estimate the true proportion of abnormal segments by more than a factor of two. We could set 
$p_d=Kd^{-1/2+\epsilon}$ for some constants $K, \epsilon>0$ and ensure that asymptotically we will consistently estimate the absence of abnormal segments (Theorem \ref{thm:1}) and the location of any abnormal segments
(Theorem \ref{thm:2}) the true proportion of abnormal segments decays at a rate that is slower than $d^{-1/2+\epsilon}$.

\section{Results}
\label{sec:results}

We call the method introduced in Sections \ref{sec:model} and \ref{sec:inference} BARD: Bayesian Abnormal Region Detector. 
We now evaluate BARD on both simulated and real CNV data. Our aim is to both investigate its robustness to different types of model mis-specification, and to compare its performance 
with a recently proposed method for analysing such CNV data.

The simulation studies we present are based on the concrete example in Section \ref{sec:cnv_example}, namely the change in mean model for Normally distributed data. For inference
we assume that the LOS distributions, $S_N$ and $S_A$, to be Negative binomial and the prior probability of a particular dimension $k$  being abnormal $p_k$ as the same for all $k = \{1,2, \hdots , d\}$.
For all the simulation studies we present we used MCEM 
on a single replicate of the simulated data set to get estimates for the hyper-parameters for the LOS distribution, but 
fixed $p_k$. 
Data for normal segments are IID standard Gaussian, and for abnormal segments data from
dimensions that are abnormal are Gaussian with variance 1 but mean $\mu$ drawn from some prior $\pi(\mu)$. Below we consider the
effect of varying the choice of prior used for simulating the data and that assumed within BARD. 
In implementing BARD we used the SRC method of resampling described in Section \ref{sec:particle_approx} with a value of $\alpha=10^{-4}$, we found this value of $\alpha$ gave a good trade off between accuracy and computational cost.


To get an explicit segmentation from BARD we use the asymmetric loss function \eqref{eq:assym_loss} with a value of $\gamma=1/3$.

As a benchmark for comparison we also analyse all data sets using the Proportion Adaptive Segment Selection procedure (PASS) from \citet{rare_copy_number}. This was implemented using an R package called {\texttt{PASS}} 
which we obtained from the authors website.   
At its most basic level the PASS method involves evaluating a test statistic for different segments of the data. After these evaluations the values of the statistic 
that exceed a certain pre-specified threshold are said to be significant and the segments that correspond to these values are the identified abnormal segments. This threshold is typically found by simulating data sets with no abnormal segments and then choosing the threshold which gives a desired type 1 error, here we take this error to be 0.05 in the simulation studies. The PASS
algorithm considers all segments that are shorter than a pre-defined length. To avoid excessive computational costs this length should be as small as possible, but at least as large as the longest abnormal segment we wish to detect (or believe exists in the data).  

We found that a run of PASS was about twice as fast as one run of BARD. In order to estimate the hyper-parameters using MCEM 
took between 5 and 20 runs of BARD.


\emph{Evaluating a segmentation}

To form a comparison between the two methods we must have some way of evaluating the quality of a particular segmentation with respect to the ground truth. We consider the three most important criteria to be the number of true and false positives and the accuracy in detecting the true positives.

We define a segment to be correctly identified or a true positive if it intersects with the true segment. With this definition in mind then finding the true/false positives is simple. To define the accuracy of an estimated segment compared to the truth it is most intuitive to measure the amount of ``overlap'' of the segments, this is captured by the dissimilarity measure $D_k$ \eqref{eq:diss_measure} defined in \citet{rare_copy_number}.

Let $\hat{ \mathbb{I}}$ be the collection of estimated intervals, the accuracy of estimating the $k^{th}$ true segment $I_k$ is given by $D_k$
\begin{align}
\label{eq:diss_measure}
  D_k = \min_{ \hat{I}_j \in \hat{\mathbb{I}}  } \left\{ 1 - \frac{ \lvert \hat{I}_j \cap I_k  \rvert }{ \sqrt{ \lvert  \hat{I}_j \rvert \lvert I_k \rvert } }  \right\}  
\end{align}
$D_k \in [0,1]$, if $D_k=0$ then an estimated interval overlaps exactly with segment $I_k$ however if $D_k =  1$ then no estimated intervals overlap with the $k^{th}$ segment, i.e.\ it hasn't been detected. Smaller values of $D$ indicate a greater overlap.

\subsection{Simulated Data from the Model}
\label{sec:sim_study}
Firstly we analysed data simulated from the model assumed by BARD. A 
soft maximum on the length of the simulated data of $n=1000$ was imposed and the number of dimensions fixed at $d=200$. 
The LOS distributions were
\begin{align*}
S_N \sim \text{NBinom}(10,0.1) \text{ and } S_A \sim \text{NBinom}(15,0.3).  
\end{align*}
Two different distributions were used to generate the altered means for the affected dimensions and we also varied $\pi_N$ 
(see Table \ref{table:res}), and for each scenario
we implemented the Bayesian method with the correct prior for the abnormal mean, and the correct chocie of $\pi_N$.
The number of affected dimensions for each abnormal segment was fixed at 4\% and we fixed $p_k$ to this value. For each scenario we
considered we generated 200 data sets.
\begin{table}[h!] 
\centering 
\begin{tabular}{l c c ccc} 
\hline\hline 
 $\mu$ & $\pi_N$ & Method & Proportion detected & Accuracy & False positives
\\ [0.5ex] 
\hline 
 &     & PASS  &  (0.66,0.70) & (0.11,0.13) & (0.68,0.93)   \\[-1ex] 
\raisebox{1.5ex}{$U(0.3,0.7)$} & \raisebox{1.5ex}{0.5} & BARD
 & (0.87,0.89) & (0.070,0.084) & (0.04,0.12)     \\[1ex] 

 &     & PASS  &  (0.66,0.71) & (0.12,0.14) & (0.90,1.19)    \\[-1ex] 
 & \raisebox{1.5ex}{0.8} & BARD
  & (0.75,0.78) & (0.081,0.093)  & (0.03,0.09)     \\[1ex] 

 &     & PASS  & (0.91,0.93)  & (0.070,0.077)  &  (0.93,1.22)   \\[-1ex] 
\raisebox{1.5ex}{$U(0.5,0.9)$} & \raisebox{1.5ex}{0.5} & BARD
 & (0.98,0.99) & (0.035,0.042) & (0.01,0.06)     \\[1ex] 

 &     & PASS  &  (0.93,0.95) & (0.069,0.076) & (0.88,1.17)    \\[-1ex] 
 & \raisebox{1.5ex}{0.8} & BARD
  & (0.95,0.97) & (0.040,0.045) & (0.00,0.04)     \\[1ex] 
\hline
\end{tabular}
\caption{Scenarios differed in the prior for $\mu$ and the value of $\pi_N$ used to simulate the data. In BARD these same priors were used for the analysis of the data. The results are based on 200 simulated data sets for each scenario and the intervals given are 95\% confidence intervals calculated using 1000 bootstrap replicates.}
\label{table:res}  
\end{table}
                                             
Results summarising the accuracy of the segmentations obtained by the two methods are shown in Table \ref{table:res}. 
BARD performed substantially better than PASS here especially with regards to the number of false positives each method found, though
this is in part because all the modelling assumptions within BARD are correct for these simulated data sets.
It is worth noting that both methods do much better when $\mu \sim U(0.5, 0.9)$ due to the stronger signal present.

\begin{table}[h!] 
\centering 
\begin{tabular}{c c c c} 
\hline\hline 
 $p_k$ & Proportion detected & Accuracy & False positives \\ [0.5ex] 
\hline
$\frac{1}{200}$ & (0.63,0.67) & (0.086,0.10) & (0.005,0.06) \\[1ex]
$\frac{4}{200}$ & (0.74,0.78) & (0.086,0.10) & (0.05,0.12) \\[1ex]
$\frac{8}{200}$ & (0.75,0.78) & (0.081,0.093) & (0.03,0.09) \\[1ex]
$\frac{12}{200}$ & (0.74,0.78) & (0.083,0.096) & (0.03,0.095) \\[1ex]
$\frac{16}{200}$ & (0.72,0.76) & (0.084,0.098) & (0.03,0.09) \\[1ex]
$\frac{20}{200}$ & (0.70,0.74) & (0.088,0.102) & (0.02,0.08)  \\[1ex]
\hline 
\end{tabular} 
\caption{The robustness of BARD under a misspecification of $p_k$ taking the prior as $\mu \sim U(0.3,0.7)$ and $\pi_N =0.8$ 
with the true value of $p_k$ being 4\%. Values of $p_k$ were varied between 0.5\% and 10\% and we simulated 200 data sets for each $p_k$.}
\label{table:robust_res}  
\end{table}

We next investigated how robust the results were to our choice for $p_k$. We just consider $\mu\sim U(0.3,0.7)$ and $\pi_N=0.8$ and
we vary our choice of $p_k$ from $0.5\%$ to $10\%$. These results are in table \ref{table:robust_res}. Whilst, as expected,
if we take $p_k$ to be the true value for the data we get the best segmentation, the results are clearly robust to mis-specification of
$p_k$. In all cases we still achieve much higher accuracy and fewer false positives than PASS. Apart from
the choice $p_k=1/200$ we also have a higher proportion of correctly detected CNVs than PASS.

We also investigated the robustness to mis-specification of the model for the LOS distribution, and for the distribution of the mean of the abnormal segments. 
We fixed the position of five abnormal segments at the following time points 200, 300, 500, 600 and 750. Additionally the segments at 200 and 750 were followed by another abnormal segment. 
Thus we have seven abnormal segments in total. The true LOS distribution for the abnormal segments are in fact Poisson with intensity randomly chosen from the set $\{ 20,25,30,35,40 \}$. 
For these abnormal segments the mean value that affected the dimensions was drawn from a Normal distribution with differing means and a fixed variance shown in Table \ref{table:fair_res}. 
The number of affected dimensions for each of the abnormal segments was also varied randomly from 3-6\% of the total number of dimensions $(d=200)$. For inference, we fixed $p_k$ to 4\% for all $k$ and 
we set the prior for the abnormal mean to be uniform on $(-0.7,-0.3) \cup (0.3,0.7)$. Our model for the LOS distribution were negative binomials, with MCEM used to estimate the
hyper-parameters of these distributions.
\begin{table}[h!] 
\centering 
\begin{tabular}{c  c ccc} 
\hline\hline 
 $\mu$ & Method & Proportion detected & Accuracy & False positives
\\ [0.5ex] 
\hline 
 
 &  PASS  &  (0.78,0.82) & (0.056,0.068) & (1.15,1.41)   \\[-1ex] 
\raisebox{1.5ex}{$N(0.8,0.4^2)$}
 & BARD  & (0.82,0.86) & (0.048,0.059) & (0.02,0.07) \\[1ex] 
 &  PASS  &  (0.74,0.78) & (0.069,0.084) & (1.05,1.33)   \\[-1ex] 
\raisebox{1.5ex}{$N(0.7,0.4^2)$}
 & BARD  & (0.78,0.82) & (0.060,0.073) & (0.01,0.07) \\[1ex] 
 &  PASS  &  (0.66,0.71) & (0.079,0.095) & (1.08,1.37)   \\[-1ex] 
\raisebox{1.5ex}{$N(0.6,0.4^2)$}
 & BARD  & (0.70,0.75) & (0.061,0.072) & (0.03,0.09) \\[1ex] 
 &  PASS  &  (0.60,0.65) & (0.089,0.11) & (1.06,1.37)   \\[-1ex] 
\raisebox{1.5ex}{$N(0.5,0.4^2)$}
 & BARD  & (0.62,0.68) & (0.075,0.093) & (0.02,0.08) \\[1ex] 
 &  PASS  &  (0.51,0.56) & (0.10,0.13) & (0.92,1.22)   \\[-1ex] 
\raisebox{1.5ex}{$N(0.4,0.4^2)$}
 &BARD  & (0.55,0.61) & (0.084,0.10) & (0.03,0.10) \\[1ex] 

\hline 
\end{tabular} 
\caption{Results based on 200 simulated data sets as we vary the distribution from which $\mu$ was simulated from but keeping the prior $\pi(\mu)$ in BARD uniform. 95\% confidence intervals for the means were calculated using 1000 bootstrap replicates.}
\label{table:fair_res}  
\end{table}

From Table \ref{table:fair_res} it can be seen that BARD still outperforms PASS especially in regards to accuracy and the number of false positives. The performance of BARD also shows that it is robust to a misspecification of both the LOS distributions and the 
distribution from which $\mu$ was drawn from as we kept the prior in BARD the same. The performance of both methods was impacted by the decreasing mean of the Normal distributions from which $\mu$ was drawn as more of them became close to zero and thus abnormal segments became indistinguishable from normal segments.


\subsection{Simulated CNV Data}
\label{sec:realistic_sim}

We now make use of the CNV data presented in the Section \ref{sec:intro}, to obtain a more realistic model to simulate data from. 
We used the PASS method to initially segment one replicate of the data, and then analysed this segmentation to obtain information about the LOS distributions 
and the distributions that generate the data in both normal and abnormal segments.

In Figure \ref{fig:realistic_sims} we plot some of the empirical data from the segmentation given by PASS. 
To simulate data sets we either fitted distributions to these quantities or sampled from their empirical distributions. 
Firstly if we consider the two LOS distributions then for normal segments, see Figure \ref{fig:n_los}, 
we found that a geometric distribution fitted the data well. For the abnormal LOS distribution we took a discrete uniform distribution on $\{1,2 , \hdots , 200\}$. 
This was partly due to us having specified a maximum abnormal segment length of 200 in the PASS method but is potentially realistic in practice as abnormal segments longer than 200 time points are unlikely to occur. 
To support this choice we plot the empirical cdf of the ordered data and a straight line which are the quantiles of the uniform distribution we propose. 
We can see that although the fit is not perfect, this is probably due to the small sample size.

Now consider the distributions that generate the actual observations, we can think of these in two parts, one of them being a distribution 
for the ``noise'' in normal segments 
(Figure \ref{fig:n_noise}) and then the mean shift parameter for the abnormal segments (Figure \ref{fig:ab_mu}). 
Up until now we have taken this noise distribution to be standard Normal, however the data suggests that in reality it has heavier tails than the Normal distribution. 
We found that a $t$-distribution with 15 degrees of freedom was a better fit to the data so we simulated from this for the noise distribution. For the mean shift parameter 
$\mu$ we took abnormal segments found by the PASS method and looked at the means of each of the dimensions and took the affected dimensions only, this gave the histogram in Figure \ref{fig:ab_mu}.
In the study we simulated $\mu$ from this empirical distribution.        

\begin{figure}[h!]
        \centering
        \begin{subfigure}[b]{0.45\textwidth}
                \includegraphics[scale=0.45]{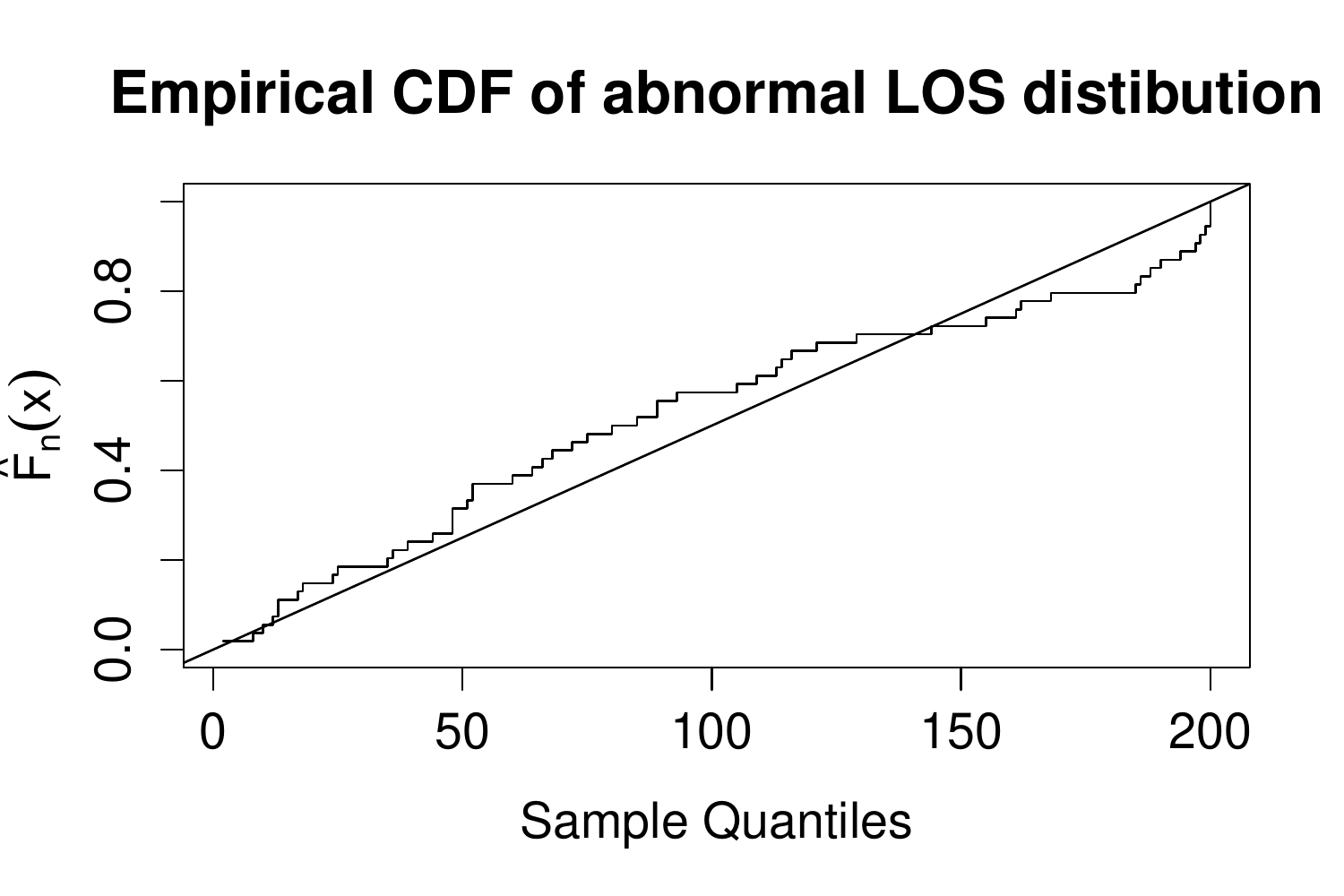}
                 \caption{}
                \label{fig:ab_los}
        \end{subfigure}
        \begin{subfigure}[b]{0.45\textwidth}
                \includegraphics[scale=0.45]{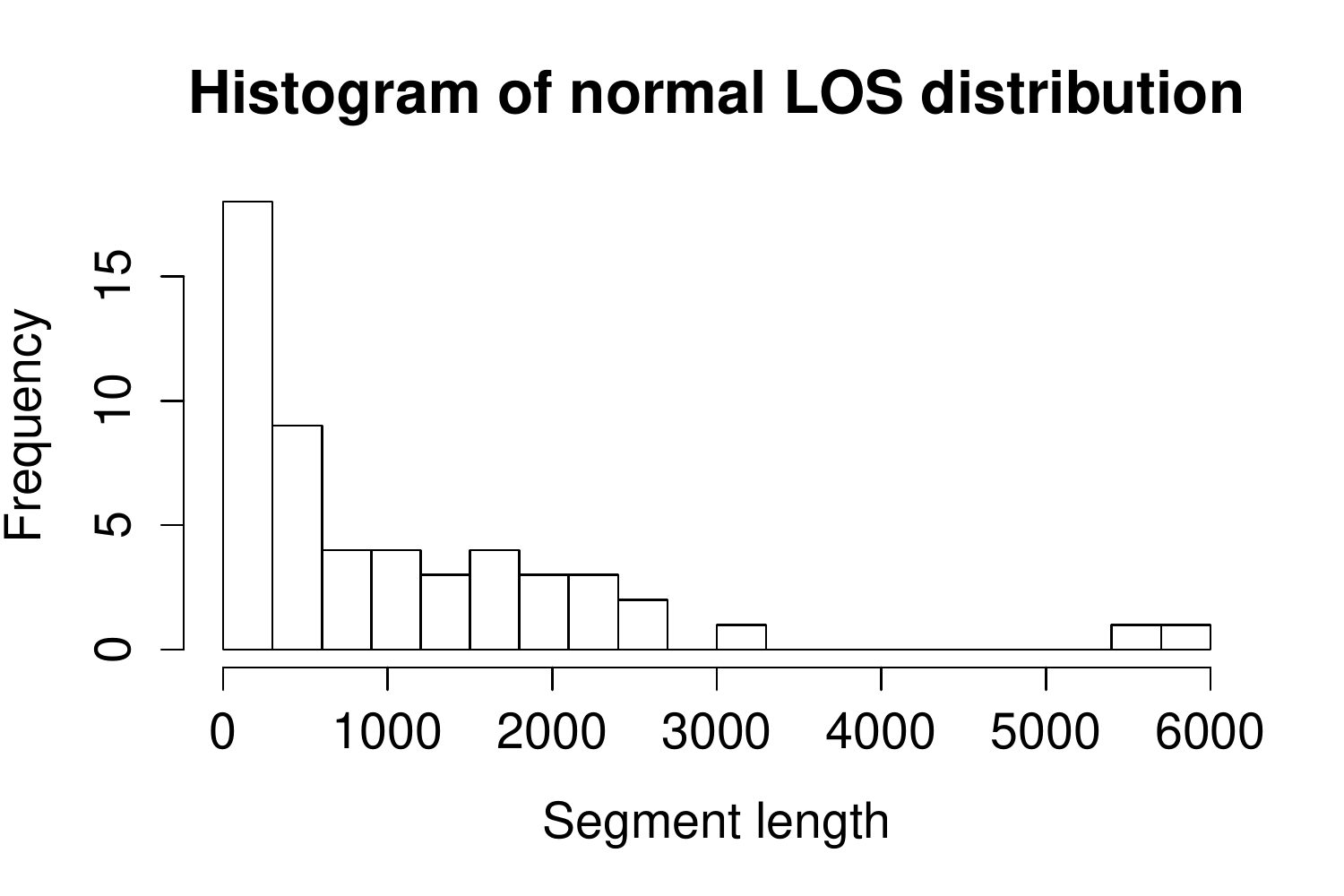}
                \caption{}
                \label{fig:n_los}
        \end{subfigure}
  
        \begin{subfigure}[b]{0.45\textwidth}
                \includegraphics[scale=0.45]{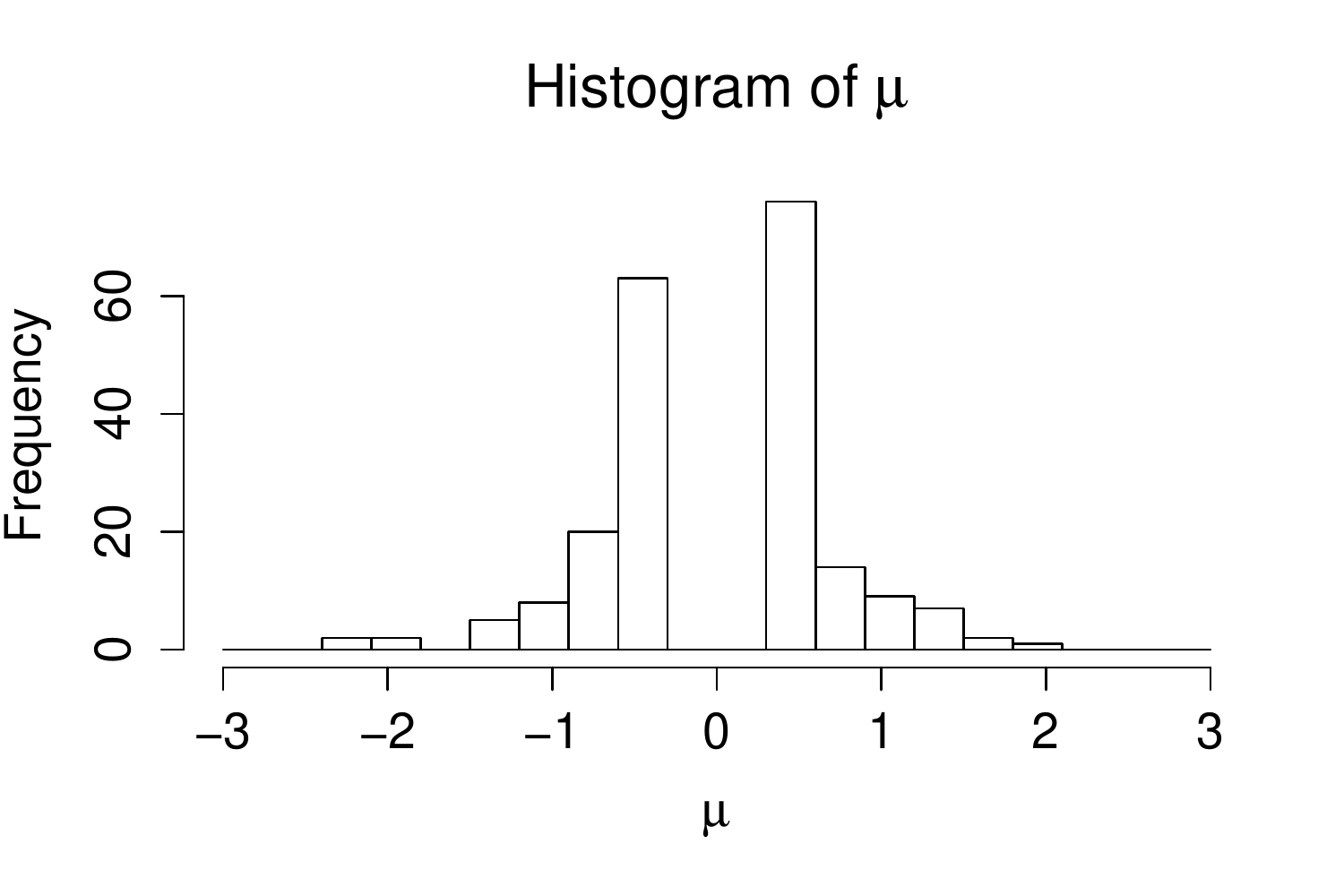}
                 \caption{}
                \label{fig:ab_mu}
        \end{subfigure}
  \begin{subfigure}[b]{0.45\textwidth}
                \includegraphics[scale=0.45]{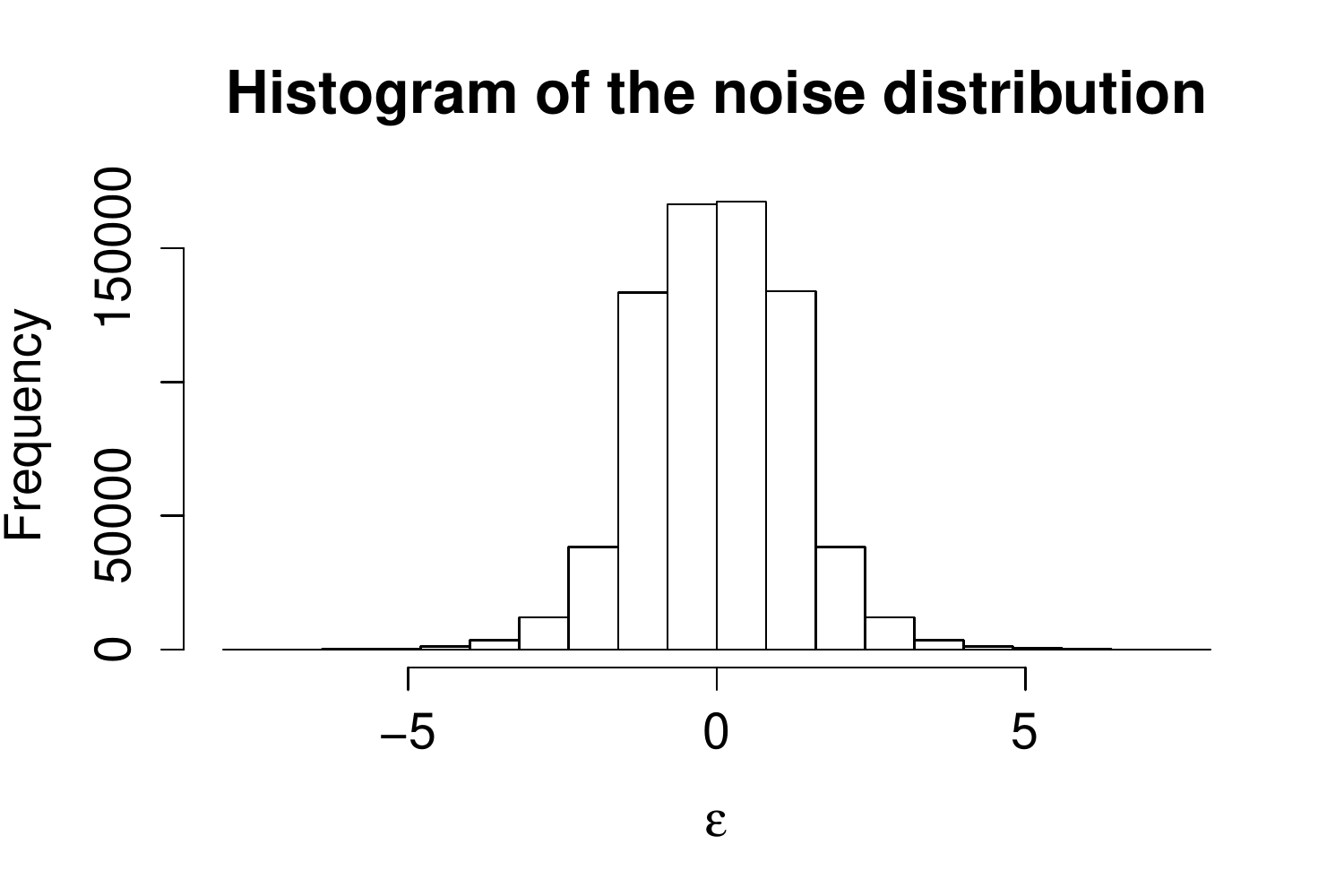}
                 \caption{}
                \label{fig:n_noise}
        \end{subfigure}
        \caption{Empirical distribution of features of the optimal segmentation of CNV data obtained using the PASS method. (a) QQ-plot of length (measured in number of observations) 
        of abnormal segments  against a Uniform distribution on $\{1,2,\ldots,200\}$;
        (b) histogram of length (measured in number of observations) of normal segments; (c) histogram of estimated mean for abnormal segments; and 
        (d) histogram of residuals.}
\label{fig:realistic_sims}
\end{figure}

Each simulated data set has length of approximately $n=20,000$ and dimension $d=50$. We also varied the proportion of affected dimensions between 4\% and 6\%. 
We simulated 40 of these data sets for each of the two scenarios and used both methods to segment them, results are given in Table \ref{table:realistic_res}.

\begin{table}[h!] 
\centering 
\begin{tabular}{c  c ccc} 
\hline\hline 
 \% of affected dimensions & Method & Proportion detected & Accuracy & False positives
\\ [0.5ex] 
\hline 
 
 &  PASS  &  (0.59,0.66) & (0.080,0.10) & (7.15,9.23)   \\[-1ex] 
\raisebox{1.5ex}{4\%}
 & BARD  & (0.61,0.69) & (0.055,0.072) & (0.08,0.38) \\[1ex] 
 &  PASS  &  (0.64,0.72) & (0.066,0.085) & (2.25,3.03)   \\[-1ex] 
\raisebox{1.5ex}{6\%}
 & BARD  & (0.71,0.78) & (0.046,0.060) & (0.10,0.43) \\[1ex]
\hline 
\end{tabular} 
\caption{Results based on 40 simulated data sets for two scenarios where the proportion of dimensions affected for each abnormal segment varied between 4\% and 6\% (of the total number of dimensions $d=50$). 95\% confidence intervals for the means were calculated using 1000 bootstrap replicates.}
\label{table:realistic_res}  
\end{table}

We can see that the proportion of correct segments identified is decreased in both methods, this is most likely due to the non-Normally distributed noise present.
However the two methods report a very different number of false positives. The performance of BARD is encouraging as it gives many fewer false positives than 
PASS even with heavier tailed observations than the standard Gaussian case. 

BARD also allows us to get an estimate of the uncertainty in the position of abnormal segments as from the posterior we can get the probability of each time point belonging to an abnormal segment. If we bin these probabilities into intervals and then find the proportion of these points that are actually abnormal we can obtain a calibration plot Figure \ref{fig:cal}. We can see from this that the model seems to be well calibrated. 

\begin{figure}[h!] 
\centering
\includegraphics[scale=0.25]{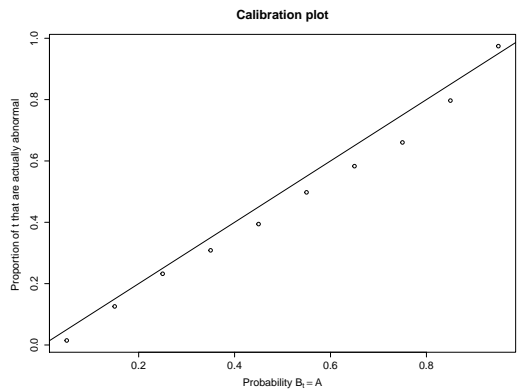}
\caption{All the time points $t$ for which the posterior probability lies in a certain interval plotted against the proportion of times $t$ lies in an abnormal segment.}
\label{fig:cal}
\end{figure}
\subsection{Analysis of CNV Data}
\label{sec:real_data}


We now apply our method to CNV data from \citet{Pinto}, a subset of which was presented in Section \ref{sec:intro} and was used to construct a model for the simulated data in Section \ref{sec:realistic_sim}. 

\citet{Pinto} undertook a detailed study of the different technologies (platforms) used to obtain the measurements and many of the algorithms currently used to call CNV's.
We chose to analyse data from the Nimblegen 2.1M platform and from chromosomes 6 and 16. For both chromosomes we have three replicate data sets, each consisting of measurements from from six genomes. 
We preprocessed the data to remove experimental artifacts, using the method described in \citet{siegmund2011},  before analysing it.
The data from chromosome 16 consisted of 59,590 measurements, and the data from chromosome 6 consisted of 126,695 measurements, for each genome.

Firstly we ran the PASS method on just the first replicate of the data from chromosome 16 and found the most significant segments.
Doing this enables us to get an estimate of the parameters for the LOS distributions to use in the 
Bayesian method without having to do any parameter inference. 
The maximum length of segment we searched over was 200 (measured in observations not base pairs) as this is greater than the largest CNV we would expect to find. This gave parameters that suggested a geometric
distribution for the length of normal segments $S_N \sim \mathrm{Geom}(0.0007)$ and the following Negative Binomial distribution for abnormal segments  $S_N \sim \mathrm{NBinom}(2,0.1)$. We used the same split uniform prior for $\mu$ as we did in Section \ref{sec:realistic_sim} namely one with equal density on the set $(-0.7,-0.3) \cup (0.3,0.7)$ and zero elsewhere. We justified the use of this form of prior which excludes values close to zero in Section \ref{sec:cnv_example} and it was shown to perform well on some realistically simulated data in Section \ref{sec:realistic_sim}. 




\begin{table}[ht]
\begin{center}
\begin{tabular}{cc|ccc|ccc}
    \hline
    \multicolumn{2}{c}{Truth} &  \multicolumn{3}{c}{PASS} &  \multicolumn{3}{c}{BARD}  \\
  \hline
Start & Length & Rep 1 & Rep 2 & Rep 3 & Rep 1 & Rep 2 & Rep 3  \\
  \hline
2619669 & 62144 & - & - & -  & - & \checkmark & \checkmark  \\

21422575 & 76266 & \checkmark & \checkmark & \checkmark  & \checkmark & \checkmark & \checkmark \\

32165010 & 456897 & \checkmark & \checkmark & \checkmark & \checkmark & \checkmark & \checkmark  \\

34328205 & 286367 & \checkmark & \checkmark & \checkmark & \checkmark & \checkmark & \checkmark  \\ 

54351338 & 28607 & \checkmark & \checkmark & -  & \checkmark & \checkmark & \checkmark  \\

70644511 & 21083 & - & \checkmark & \checkmark  & - & \checkmark & \checkmark \\
    \hline
\end{tabular}
\end{center}
\caption{Known CNV's from HapMap found by either method when analysing different replicates of data from chromosome 16. Ticks indicate whether the particular segment was detected or not.}
\label{tab:cnvchr16}
\end{table}


\begin{table}[h!]
\begin{center}
\begin{tabular}{cc|ccc|ccc}
    \hline
    \multicolumn{2}{c}{Truth} &  \multicolumn{3}{c}{PASS} &  \multicolumn{3}{c}{Bayesian}  \\
  \hline
Start & Length & Rep 1 & Rep 2 & Rep 3 & Rep 1 & Rep 2 & Rep 3  \\
  \hline
202353 & 37484 & -  & - &  - & \checkmark & \checkmark & -  \\
 
243700 & 80315 & \checkmark & \checkmark & \checkmark & \checkmark & \checkmark & \checkmark   \\ 

29945167 & 12079 & \checkmark  & \checkmark &  - & - & - & -  \\ 

31388080 &  61239 & \checkmark & - & -  & \checkmark & - & - \\ 

32562253 & 117686 &  \checkmark  & - & \checkmark  & - & - & - \\ 

32605094 & 74845 & -  & \checkmark & - & \checkmark & \checkmark & \checkmark  \\
 
32717276 & 22702 & \checkmark & - & -  & \checkmark & \checkmark & \checkmark  \\ 

74648953 & 9185 & \checkmark & \checkmark & - & \checkmark & \checkmark & \checkmark  \\ 

77073620 & 10881 & - & \checkmark & - & \checkmark & \checkmark & \checkmark  \\ 

77155307 & 781 & -  & - & -  & \checkmark & - & -  \\ 

77496587 & 12936 & - & - & -  & \checkmark & \checkmark & \checkmark  \\ 

78936990 & 18244 & \checkmark & \checkmark & \checkmark & \checkmark & \checkmark & \checkmark \\ 

103844669 & 24085 & \checkmark  & \checkmark  & \checkmark  & \checkmark  & \checkmark & \checkmark  \\ 

126225385 & 3084 & \checkmark  & \checkmark  & - & - & - & \checkmark   \\ 

139645437 & 3392 & -  & - & -  & \checkmark  & - & -  \\ 

165647807 & 4111 & -  & - & -  & \checkmark  & - & \checkmark   \\
    \hline
\end{tabular}
\end{center}
\caption{Known CNV's from HapMap found by either method when analysing different replicates of data from chromosome 6. Ticks indicate whether the particular segment was detected or not.}
\label{tab:cnvchr6}
\end{table}

\begin{table}[h!] 
\centering 
\begin{tabular}{c  c ccc} 
\hline\hline 
 Chromosome & Method & Rep 1 v 2 & Rep 1 v 3 & Rep 2 v 3
\\ [0.5ex] 
\hline 
 
 & PASS & 0.474  & 0.709  & 0.522  \\[-1ex]
 
\raisebox{1.5ex}{6}

 & BARD & 0.495 & 0.457 & 0.416  \\[1ex] 

 & PASS & 0.478 & 0.507 & 0.388    \\[-1ex] 

\raisebox{1.5ex}{16}

 & BARD & 0.426 & 0.467 & 0.682  \\[1ex]
\hline 
\end{tabular} 
\caption{The average consistency measured using the dissimilarity measure for found CNV's between replicates and methods. A lower value 
indicates the inferred segmentations for the two replicates were more similar.}
\label{tab:cnv_consistency}  
\end{table}
For both chromosomes we analysed the three replicates separately. Ideally we should infer exactly the same segmentation for each of the replicate data sets. 
Due to the large amount of noise present in the data this does not happen.  However we would expect that a ``better'' method would be more consistent across the three replicates, and we use 
the consistency of the inferred segmentations across the replicates as a measure of accuracy. 

We can also use data from the HapMap project to validate some of the CNV's we found to those known experimentally or which have been called by other authors. 
A list containing these known CNV's by chromosome and sample can be found at \texttt{http://hapmap.ncbi.nlm.nih.gov/}. 
These validated segments suggest that about 1\% of chromosome 16 is abnormal. 

To make comparisons between BARD and PASS fair we implemented both of these methods so that 
they identified the same proportion, $4\%$, of the chromosome as being abnormal. For BARD this involved choosing $\gamma$ in the loss function \eqref{eq:assym_loss} appropriately and for PASS selecting the most significant segments that give us a total of 4\% abnormal time points. We then tested these against the validated CNV's.


The results for chromosome 16 are contained in Tables~\ref{tab:cnvchr16} and \ref{tab:cnv_consistency}; and those for chromsome 6 in Tables~\ref{tab:cnvchr6} and \ref{tab:cnv_consistency}. Tables~\ref{tab:cnvchr16} and \ref{tab:cnvchr6}
list the known CNV regions that were detected by one or both methods for at least one replicate, whilst Table~\ref{tab:cnv_consistency} gives summaries of the consistency of the inferred
segmentations across replicates.

The results show that BARD is more successful at detecting known CNV regions than PASS. In total BARD found 6 CNV regions on chromosome 16 for 
at least one replicate, and 14 for chromosome 6, while PASS managed 5 and 11 respectively. For the measures of consistency across the different
replicates, shown in Table \ref{tab:cnv_consistency}, BARD performed better for 4 of the 6 pairs.

\section{Discussion}
\label{sec:discussion}

In this paper we have developed novel methodology to detect abnormal regions in multiple 
time series. Firstly we developed a general model for this type of problem including length of 
stay distributions and marginal likelihoods for normal and abnormal segments. We then derived 
recursions that could be used to calculate the posterior of interest and showed how to
obtain  iid samples from an accurate approximation to this posterior in a way that
scales linearly with the length of series. 

The resulting algorithm, BARD, 
was then compared in several simulation studies and some real data to another competing 
method PASS. These results showed that BARD was consistently more accurate than the PASS
benchmark on several important criteria for all of the data sets we considered.


The novelty of our method comes from being able to accurately and efficiently perform 
Bayesian inference for large and high dimensional data sets of this type thus allowing 
us to quantify uncertainty in the location of abnormal segments. Before this with other 
methods such as PASS this quantification of uncertainty has not been possible.


Whilst we have focused on changes in mean from some baseline level,
our method could easily be adapted to any model which specifies some normal behaviour 
and abnormal behaviour. The only restrictions we place on this is the ability to calculate 
marginal likelihoods for both types of segment. 
The only potential bottleneck would be in the calculation of the abnormal marginal 
likelihoods as this involves integration over a prior for the parameter(s) which cannot 
be done analytically, and for higher dimensional parameters would be computationally intensive.
 
R code to run the BARD method is available at the first authors website. \url{http://www.lancaster.ac.uk/pg/bardwell/Work.html}. 
The real CNV data we analysed in Section \ref{sec:real_data} is available publicly and can be downloaded from the GEO accession website \url{http://www.ncbi.nlm.nih.gov/geo/query/acc.cgi?acc=GSE25893}.

{\bf Acknowledgements}  We thank Idris Eckley for helpful comments and discussions. This research was supported
by EPSRC grant EP/K014463/1. Bardwell gratefully acknowledges funding from EPSRC and British Telecom via the STOR-i Centre
for Doctoral Training.

\appendix
\appendix

\section{Lemmas for Proof of Theorem \ref{thm:1}}


Throughout this and the following section, we will assume the data is generated from the model detailed in Section \ref{asymptotics}.

We define part of the ratio in \eqref{eq:P1} as $X_k(\mu)$
\begin{align*}
  X_k(\mu) = \exp \left\{ \mu \sum_{u=t}^{s} \left( Y_{k,u} - \frac{\mu}{2} \right)    \right\}.
\end{align*}
The random variable $X_k(\mu)$ is log-normally distributed with different parameters depending on whether the sequence is normal or abnormal for that segment. In the normal segment case it is log-normal 
with parameters $-\mu^2(s-t+1)/2$ and $\mu^2(s-t+1)$, 
\begin{align*}
 \textrm{with } \mathbb{E}X_k(\mu) = 1.
\end{align*}
For this case we will further define the $m$th central moment of $X_k(\mu)$ to be $C_m(\mu)$
\begin{align*}
  C_m(\mu) = \mathbb{E} \left[ \left( X_k(\mu) - \mathbb{E}X_k(\mu) \right)^m \right].
\end{align*}

Finally  we denote the log of the product over the $d$ terms in \ref{eq:P1} as $S_d(\mu)$, taking the logarithm makes this become a sum over all the time-series
\begin{align*}
S_d(\mu) = \sum_{k=1}^{d} \log( 1 + p_d( X_k(\mu) - 1 ) ).
\end{align*}
We now go on to prove several lemmas about $S_d(\mu)$ for both normal  segments which will aid us in proving Theorems \ref{thm:1}.


\begin{lemma}[Normal segment moment bounds]
\label{norm_seg_bound}
Assume we have a normal segment then 
\begin{align}
\begin{split}
  \mathbb{E}S_d(\mu) \leq - \frac{1}{2}C_2(\mu)dp_d^2 + \frac{1}{3}C_3(\mu)dp_d^3 \\
  \mathbb{E} \left[ \left( S_d(\mu) -  \mathbb{E}S_d(\mu) \right)^{2k} \right] \leq K_{k}(\mu) d^k p_d^{2k}
\end{split}
\end{align}
 where $C_m(\mu)$ is the $m$th central moment of $X_m(\mu)$, and $K_k(\mu)>0$ does not depend on $d$. 
\end{lemma}
\begin{proof}
Writing out the expectation of $S_d(\mu)$ gives
  \begin{align}
    \mathbb{E}S_d(\mu) = \sum_{k=1}^{d} \mathbb{E} \left[ \log( 1 + p_d( X_k(\mu) - 1 ) ) \right]
  \end{align}
then we use the inequality $\log(1+x) \leq x - \frac{x^2}{2} + \frac{x^3}{3}$ for $x > 0$. So
\begin{align*}
   \sum_{k=1}^{d} \mathbb{E} \left[ \log( 1 + p_d( X_k(\mu) - 1 ) ) \right] &\leq \sum_{k=1}^{d} \mathbb{E} \left[ p_d( X_k(\mu) - 1 ) - \frac{p_d^2(X_k(\mu)-1)^2}{2} + \frac{p_d^3(X_k(\mu)-1)^3}{3} \right] \\
&= \sum_{k=1}^{d} -p_d^2 \frac{ \mathbb{E} \left[ (X_k(\mu)-1)^2 \right] }{2} + p_d^3 \frac{\mathbb{E} \left[ (X_k(\mu)-1)^3 \right]}{3} \\
&= - \frac{1}{2}C_2(\mu)dp_d^2 + \frac{1}{3}C_3(\mu)dp_d^3.
\end{align*}

Now to derive the second inequality we consider $S_d(\mu)-\mathbb{E}S_d(\mu)$
\begin{align*}
  S_d(\mu) - \mathbb{E}S_d(\mu) &= \sum_{i=1}^{d}\left[ Z_i(\mu) - \mathbb{E}Z_i(\mu)\right]  
 = \sum_{i=1}^{d} \bar{Z}_i(\mu),
\end{align*}
where $\bar{Z}_i(\mu)=Z_i(\mu)-\mathbb{E}Z_i(\mu)$.
Writing this in terms of the centered random variables $\bar{Z}_i(\mu)$ is advantageous as when we consider raising the sum to the $2k$th power any term including a unit power of $\bar{Z}_i(\mu)$ 
vanishes by independence as $\mathbb{E}\bar{Z}_i(\mu) = 0$. Define 
\[
 \mathcal{I}_{d,k}=\left\{ (j_1,\ldots,j_d): j_i\in\{0,2,3,\ldots,2k\} \mbox{ for } i=1,\ldots,d \mbox{ and } \sum_{i=1}^d j_i=2k\right\},
\]
the set of non-negative integer vectors of length $d$, whose entries sum to $2k$, and that have no-entry that is equal to 1. For $\mathbf{j}\in \mathcal{I}_{d,k}$, let $n_{\mathbf{j}}$ be the number of terms in the
expansion of $(\sum_{i=1}^{d} \bar{Z}_i(\mu))^{2k}$ which have powers $j_i$ for $\bar{Z}_i(\mu))$. Thus
\begin{eqnarray*}
 \mathbb{E}\left[ \left( S_d(\mu) - \mathbb{E}S_d(\mu) \right)^{2k} \right]&=& \mathbb{E}\left[\left(  \sum_{i=1}^{d} \bar{Z}_i(\mu) \right)^{2k} \right]\\
  &=& \sum_{\mathbf{j}\in\mathcal{I}_{d,k}} n_{\mathbf{j}} \prod_{i=1}^d \mathbb{E}\left(\bar{Z}_i(\mu)^{j_i} \right) \\
  &\leq& \mathbb{E}\left(\bar{Z}_1(\mu)^{2k} \right)\sum_{\mathbf{j}\in\mathcal{I}_{d,k}} n_{\mathbf{j}}.
\end{eqnarray*}
Using $|\log(1+x)|\leq |x|+x^2/2$, we can bound $\mathbb{E}(\bar{Z}_1^{2k})$ by $A_k(\mu)p_d^{2k}$, where $A_k(\mu)$ will depend only on the the first $2k$ moments of $X_k(\mu)$, but not on $p_d$.
Finally note that each term in $\mathcal{I}_{d,k}$ can only involve vectors with at most $k$ non-zero components. For a term with $l$ non-zero-components there will be $O(d^l)$ possible choices
for which components are non-zero. Hence we have that
\[
 \sum_{\mathbf{j}\in\mathcal{I}_{d,k}} n_{\mathbf{j}} \leq B_k d^k ,
\]
for some constant $B_k$ that does not depend on $d$. Thus we have the required result, with $K_k(\mu)=A_k(\mu)B_k$.
\end{proof}

\begin{lemma}[Probability bound]
\label{prob_bound}
Fix $\mu$  and assume $p_d\rightarrow 0$ as $d\rightarrow \infty$. For a normal segment we have that there exists $D_k(\mu)>0$ such that for sufficiently large $d$ 
\begin{align}
  \Pr \left( S_d(\mu) \geq -\frac{1}{4}C_2(\mu) dp_d^2 \right) \leq \frac{D_k(\mu)}{d^k p_d^{2k}}.
\end{align}  
\end{lemma}
\begin{proof}
We first bound the probability by the absolute value of the centered random variable and then use Markov's inequality with an even power of the form $2k$
\begin{align*}
  \Pr\left( S_d(\mu) \geq - \frac{1}{4}C_2(\mu)dp_d^2 \right) &\leq \Pr \left( \left\lvert S_d(\mu) - \mathbb{E}S_d(\mu) \right\rvert \geq \frac{1}{4}C_2(\mu)dp_d^2 -\frac{1}{3}C_3(\mu)dp_d^3 \right) \\
&\leq \frac{ \mathbb{E}\left[ \left( S_d(\mu) - \mathbb{E}S_d(\mu) \right)^{2k} \right] }{ (\frac{1}{4}C_2(\mu)dp_d^2-\frac{1}{3}C_3(\mu)dp_d^3)^{2k} }. 
\end{align*}
For $d$ sufficiently large that $2C_3(\mu) p_d< C_2(\mu)$, we have
\[
 \frac{1}{4}C_2(\mu)dp_d^2-\frac{1}{3}C_3(\mu)dp_d^3 > \frac{1}{12}C_2(\mu)dp_d^2.
\]
Now using the result from Lemma \ref{norm_seg_bound} we can replace the $2k$th centered moment by the bound we obtained above. Thus for sufficiently large $d$, 
\begin{align*}
  \Pr\left( S_d(\mu) \geq - \frac{1}{4}C_2(\mu)dp_d^2 \right) \leq \frac{K_k(\mu)d^kp_d^{2k}}{ (\frac{1}{12}C_2(\mu)dp_d^2)^{2k} } 
\end{align*}
So the result holds with $D_k(\mu)=K_k(\mu)[C_2(\mu)/12]^{-2k}$.
\end{proof}

\begin{lemma}[Lower bound for the second derivative of $S_d(\mu)$]
\label{second_derivative}
We have that 
\begin{align*}
 \frac{\mbox{d}^2 S_d(\mu)}{ \mbox{d} \mu^2} \geq -d(s-t+1)  
\end{align*}
\end{lemma}
\begin{proof}
Firstly note that 
\begin{align*}
  \frac{\mbox{d} X_k(\mu) }{ \mbox{d} \mu} = \left( \sum_{u=t}^{s} y_{k,u} - \mu(s-t+1) \right)X_k(\mu).
\end{align*}

Now differentiating $S_d(\mu)$ twice  
\begin{align*}
\frac{\mbox{d} S_d(\mu) }{ \mbox{d} \mu} &= \sum_{k=1}^{d} \frac{p_d \left( \sum_{u=t}^{s}y_{k,u} - \mu (s-t+1) \right) X_k(\mu) }{1 + p_d ( X_k(\mu) -1 )} \\
\frac{\mbox{d}^2 S_d(\mu) }{ \mbox{d} \mu^2} &= \sum_{k=1}^{d} \frac{ -p_d(s-t+1)X_k(\mu) + p_d\left( \sum_{u=t}^{s} y_{k,u}-\mu(s-t+1) \right)^2 X_k(\mu) }{ 1 + p_d( X_k(\mu) - 1 )  }  \\
& - \left( \sum_{u=t}^{s} y_{k,u}-\mu(s-t+1) \right)^2 \left( \frac{ p_d X_k(\mu) }{ 1 + p_d(X_k(\mu) - 1 ) } \right)^2 
\end{align*}
Let 
\begin{align*}
  Q_k = \frac{p_dX_k(\mu) }{1 + p_d ( X_k(\mu) -1 )}
\end{align*}
and $0 \leq Q_k \leq 1$ as $1-p_d > 0$ (or $p_d < 1$). Thus the second derivative 
\begin{align*}
  \frac{\mbox{d}^2 S_d(\mu) }{\mbox{d} \mu^2} &= \sum_{k=1}^d \left[ -(s-t+1)Q_k + \left( \sum_{u=t}^{s}y_{k,u} - \mu(s-t+1) \right)^2 (Q_k-Q_k^2) \right] \\
& \geq \sum_{k=1}^d -(s-t+1)Q_k \geq -d(s-t+1) 
\end{align*}
has the required lower bound.
\end{proof}

\begin{lemma}[Detection of normal segments]
\label{case_1}
Let $\pi(\mu)$ be a density function with support $[a,b]$ with $a>0$ and $b<\infty$, and assume $1/p_d = O(d^{\frac{1}{2}-\epsilon})$ for some $\epsilon>0$. For a normal segment $[t,s]$, 
\begin{align}
\label{eq:normseg}
  \int \left\{ \prod_{k=1}^d \frac{P_{A,k}(t,s;\mu)}{P_{N,k}(t,s)}\right\} \pi(\mu) \mbox{d}\mu \rightarrow 0
\end{align}
in probability as $d \rightarrow \infty$.
\end{lemma}
\begin{proof}
Define $C_2=\min_{\mu\in[a,b]} C_2(\mu)$, and for a given $d$, $M_d$ to be the smallest integer that is greater than
\[
\frac{(b-a)\sqrt{s-t+1}}{p_d\sqrt{C_2}}. 
\]
Define $\Delta_d=(b-a)/M_d$. Now we can partition $[a,b]$ into $M_d$ intervals of the form $[\mu_{i-1} , \mu_{i}]$ for $i=1,\ldots,M_d$, where $\mu_{i}=a+i\Delta_d$.
Then the left-hand side of \eqref{eq:normseg} can be rewritten as
\begin{align*}
  \sum_{i=1}^{M_d} \int_{\mu_{i-1}}^{\mu_{i}}  \left\{  \prod_{k=1}^{d} \left[ 1 + p_d(X_k(\mu) - 1) \right] \right\} \pi(\mu) \mbox{d}\mu. 
\end{align*}
Remember that $S_d(\mu)=\sum_{k=1}^d \log [1 + p_d(X_k(\mu) - 1)]$.
Let $E_d$ be the event that 
\[
 S_d(\mu)< -\frac{1}{4}C_2 d p_d^2, \mbox{for all $\mu=\mu_i$, $i=0,\ldots,M_d$}.
\]
If this event occurs then
\[
 \max_{\mu\in[a,b]} S_d(\mu) < -\frac{1}{4}C_2 d p_d^2+ \Delta_d^2 d(s-t+1)/8,
\]
as using Lemma \ref{second_derivative} we can bound $S_d(\mu)$ on each interval $[\mu_i,\mu_{i+1}]$ by a quadratic with second derivative $-d(s-t+1)$ and which takes values $-\frac{1}{4}C_2 d p_d^2$ at
the end-points. 

Now by definition of $\Delta_d$,
\[
 -\frac{1}{4}C_2 d p_d^2+ \Delta_d^2 d(s-t+1)/8 < -\frac{1}{4}C_2 d p_d^2 + \frac{1}{8}C_2 d p_d^2 \rightarrow -\infty
\]
as $d \rightarrow \infty$ because $dp_d^2\rightarrow \infty$ under our assumption on $p_d$. Thus to prove the Lemma we need only show that event $E_d$ occurs with probability 1 as $d\rightarrow \infty$.

We can bound the probability of $E_d$ not occurring using Lemma \ref{prob_bound}. For any integer $k>0$ we have that the probability $E_d$ does not occur is
\begin{eqnarray*}
\sum_{i=1}^{M_d+1} \Pr\left(S_d(\mu_i)\geq -\frac{1}{4}C_2dp_d^2\right) &\leq&  \sum_{i=1}^{M_d+1} \Pr\left(S_d(\mu_i)\geq -\frac{1}{4}C_2(\mu_i)dp_d^2\right) \\
&\leq& \sum_{i=1}^{M_d+1} \frac{D_k(\mu_i)}{d^kp_d^{2k}} \\
&\leq& (M_d+1) \max_{\mu\in [a,b]} \frac{D_k(\mu)}{d^kp_d^{2k}}.
\end{eqnarray*}
Here $D_k(\mu)$ is defined in Lemma \ref{prob_bound}. It is finite for any $\mu$, and hence  $\max_{\mu\in [a,b]} D_k(\mu)$ is finite.

Now $M_d=O(p_d^{-1})$, so we have that the above probability is $O(d^{-k}p_d^{-2k-1})=O(d^{1/2-(2k+1)\epsilon})$. So by choosing $k>1/(4\epsilon)$ this is $O(d^{-\epsilon})$ which tends to 0 as required.
\end{proof} 

\section{Lemmas for Proof of Theorem \ref{thm:2}}

We use the same notation as in Section \ref{thm:1}. However, we will now consider an abnormal segment from positions $t$ to $s$. Let $\alpha_d$ denote the proportion of sequences that
are abnormal, and $\mu_0$ the mean. The observations in this segment come from a two component mixture. With probability $\alpha_d$ they are normally distributed with mean $\mu_0$ and variance 1;
otherwise they have a standard normal distribution. It is straightforward to show that for such an abnormal segment,
\begin{align}
  \mathbb{E}X_k(\mu) = (1-\alpha_d) + \alpha_de^{\mu\mu_0(s-t+1)}. \label{eq:m1}
\end{align}

\begin{lemma}[Abnormal segments, expectation and variance]
\label{abnormal_mean_var}
Assume we have an abnormal segment $[t,s]$ with the mean of affected dimensions being $\mu_0$.
Let $f(\mu)$ be a density function with support  $A \subset \mathbb{R} $ then
\begin{align*}
  \mathbb{E} \left[  \int_{A} S_d(\mu) f(\mu) \mbox{d}\mu \right] \geq D_1(\mu) dp_d \\
  \mathrm{Var}\left( \int_{A} S_d(\mu) f(\mu) \mbox{d}\mu \right) \leq D_2(\mu) dp_d^2+o(dp_d^2)
\end{align*}
 with
 \begin{eqnarray}
  D_1(\mu)&=&\min_{\mu \in A} \left( \mathbb{E}[X_k(\mu) -1] - \frac{p_d}{2} \mathbb{E}[ (X_k(\mu)-1)^2 ]  \right) \label{eq:l1}\\
 &=& \min_{\mu \in A} \left[ \alpha_d( e^{\mu\mu_0(s-t+1)} - 1 ) -\frac{p_d}{2}\left( e^{\mu^2(s-t+1)} - 1 \right) - \frac{\alpha_dp_dC(\mu)}{2}  \right]\label{eq:l2} \\
 C(\mu) &=& e^{\mu^2(s-t+1)}( e^{2\mu\mu_0(s-t+1)} - 1 ) - 2(e^{\mu\mu_0(s-t+1)} - 1) \nonumber
 \end{eqnarray}
 and
\[
 D_2(\mu)=\max_{\mu \in A} \mathbb{E}\left[(X_k(\mu) - 1)^2\right].
\]
\end{lemma}
\begin{proof}
As $S_d(\mu)$ is the sum of $d$ iid terms we can rewrite the expectation and variance with a single term
\begin{align*}
  \mathbb{E} \left[  \int_{A} S_d(\mu) f(\mu) \mbox{d}\mu \right] &= d \mathbb{E}\left[ \int_{A} \log(1+p_d( X_k(\mu) - 1 )) f(\mu) \mbox{d}\mu \right] \\
\mathrm{Var}\left(  \int_{A} S_d(\mu) f(\mu) \mbox{d}\mu  \right) &= d \mathrm{Var}\left( \int_{A} \log(1+p_d( X_k(\mu) - 1 )) f(\mu) \mbox{d}\mu   \right).
\end{align*}
Now as $\log(1+x)>x-x^2/2$,
\begin{align*}
  \mathbb{E}\left[ \int_{A} \log(1+p_d( X_k(\mu) - 1 )) f(\mu) \mbox{d}\mu \right] &\geq \mathbb{E}\left[ \int_{A} \left( p_d( X_k(\mu) - 1 ) - \frac{p_d^2( X_k(\mu) - 1 )^2}{2} \right) f(\mu) \mbox{d}\mu \right] \\
&= p_d \int_{A} \left( \mathbb{E}[X_k(\mu) -1] - \frac{p_d}{2} \mathbb{E}[ (X_k(\mu)-1)^2 ]  \right) f(\mu) \mbox{d}\mu,
\end{align*}
which gives (\ref{eq:l1}). We then obtain (\ref{eq:l2}) by using (\ref{eq:m1}) and a similar calculation for the variance of $X_k(\mu)$.

We now consider the variance, which is bounded by the second moment. Using $|\log(1+x)|\leq |x|+x^2/2$ we have
\begin{eqnarray*}
 \lefteqn{\mathrm{Var}\left(  \int_{A} \log ( 1 + p_d(X_k(\mu) - 1) ) f(\mu) \mbox{d}\mu \right) \leq \mathbb{E} \left[ \left(  \int_{A}  \log( 1 + p_d(X_k(\mu) - 1) ) f(\mu) \mbox{d}\mu \right)^2   \right]} \\
&\leq& \mathbb{E} \left[  \int_{A}  \left\{\log( 1 + p_d(X_k(\mu) - 1))\right\}^2 f(\mu) \mbox{d}\mu \right]\\
& \leq & \mathbb{E} \left[ \int_{A} \left\{p_d^2(X_k(\mu) - 1)^2+p_d^3|X_k(\mu) - 1|^3+\frac{p_d^4}{4} (X_k(\mu) - 1)^4 \right\}f(\mu) \mbox{d}\mu\right] \\
&\leq & \max_{\mu\in A} \mathbb{E}\left\{p_d^2(X_k(\mu) - 1)^2\right\} \int_A f(\mu)\mbox{d}\mu+o(p_d^2),
\end{eqnarray*}
which gives the required bound for the variance.
\end{proof}

\begin{lemma}[Detection of abnormal segments]
\label{case_2}
Assume that we have an abnormal segment $[t,s]$. Let $\alpha_d$ be the probability of a sequence being abnormal and the mean of the abnormal observations be $\mu_0$, with $p_d = o(1)$. Assume that there exists a set $A$ such that for all $\mu \in A$ we have
\[
 \lim_{d\rightarrow \infty} \alpha_d\left( e^{\mu\mu_0(s-t+1)}- 1 \right)-\frac{p_d}{2}\left( e^{\mu^2(s-t+1)} - 1 \right)>\delta,
\]
and $\int_A \pi(\mu)\mbox{d}\mu>\delta'$, for some $\delta,\delta'>0$. 
If $dp_d^2\rightarrow \infty$ as $d\rightarrow \infty$ then
\begin{align}
\label{eq:abseg}
  \int \left\{ \prod_{k=1}^d \frac{P_{A,k}(t,s;\mu)}{P_{N,k}(t,s)}\right\} \pi(\mu) \mbox{d}\mu \rightarrow \infty
\end{align}
in probability as $d \rightarrow \infty$.
\end{lemma}
\begin{proof}
If we restrict the integral in \eqref{eq:abseg} to one over $A \subset \mathbb{R}$ we get a lower bound. Then rewriting the ratio in \eqref{eq:abseg}, using \eqref{eq:P1}, in terms of $X_k(\mu)$ we get  
\begin{align*}
\int \left\{  \prod_{k=1}^{d} \left[ 1 + p_d(X_k(\mu) - 1) \right] \right\} \pi(\mu) \mbox{d}\mu &\geq \int_{A} \left\{  \prod_{k=1}^{d} \left[ 1 + p_d(X_k(\mu) - 1) \right] \right\} \pi(\mu) \mbox{d}\mu. 
\end{align*}
If we consider the logarithm of the above random variable and use Jensen's inequality we get a lower bound
\begin{align*}
  \log \left( \int_{A} \left\{  \prod_{k=1}^{d} \left[ 1 + p_d(X_k(\mu) - 1) \right] \right\} \pi(\mu) \mbox{d}\mu \right) 
&\geq   \int_{A} \left\{ \sum_{k=1}^{d}  \log( 1 + p_d( X_k(\mu) - 1 ) ) \right\} \pi(\mu) \mbox{d}\mu \\
&=  \int_{A} S_d(\mu) \pi(\mu) \mbox{d}\mu.
\end{align*}
Then if we can show this random variable goes to $\infty$ as $d \rightarrow \infty$ the original random variable has the same limit. 
Let $T_d =  \int_{A} S_d(\mu)\pi(\mu) \mbox{d}\mu$. Using Lemma \ref{abnormal_mean_var}, we have
\[
 \mbox{E}(T_d)> \log(\delta')+\delta dp_d, 
\]
and for sufficiently large $d$ there exists a constant $C$ such that
\[
 \mbox{Var}(T_d)< C dp_d^2.
\]
So by Chebyshev's inequality
\begin{align*}
  \Pr( T_d \leq \log(\delta')+\delta dp_d - dp_d^2 ) \leq \Pr( \lvert T_d -\mathbb{E}T_d \rvert \geq dp_d^2 ) \leq \frac{ \mathrm{Var}(T_d) }{d^2p_d^4} < \frac{C}{dp_d^2}.
\end{align*}
Thus $T_d \rightarrow \infty$  in probability as $d\rightarrow \infty$, which implies (\ref{eq:abseg}).
\end{proof}

\bibliography{bibliog}

\begin{thebibliography}{}

\bibitem[Barry and Hartigan, 1992]{Barry/Hartigan:1992}
Barry, D. and Hartigan, J.~A. (1992).
\newblock Product partition models for change point problems.
\newblock {\em The Annals of Statistics}, 20(1):260--279.

\bibitem[Berger, 1985]{decision_theory_berger}
Berger, J. (1985).
\newblock {\em Statistical decision theory and Bayesian analysis}.
\newblock Springer series in statistics. Springer, New York, NY [u.a.], 2. ed
  edition.

\bibitem[Cox, 1962]{cox1962renewal}
Cox, D. (1962).
\newblock {\em Renewal Theory}.
\newblock Methuen's monographs on applied probability and statistics. Methuen.

\bibitem[Fearnhead, 2006]{Fearnhead2006}
Fearnhead, P. (2006).
\newblock Exact and efficient {B}ayesian inference for multiple changepoint
  problems.
\newblock {\em Statistics and Computing}, 16(2):203--213.

\bibitem[Fearnhead and Liu, 2007]{FearnheadLiu2007}
Fearnhead, P. and Liu, Z. (2007).
\newblock On-line inference for multiple changepoint problems.
\newblock {\em Journal of the Royal Statistical Society B}, 69:589--605.

\bibitem[Fearnhead and Vasileiou, 2009]{FearnheadVasileiou2009}
Fearnhead, P. and Vasileiou, D. (2009).
\newblock Bayesian analysis of isochores.
\newblock {\em Journal of the American Statistical Association},
  104(485):132--141.

\bibitem[Frick et~al., 2014]{Frick:2014}
Frick, K., Munk, A., and Sieling, H. (2014).
\newblock Multiscale change point inference.
\newblock {\em Journal of the Royal Statistical Society: Series B (Statistical
  Methodology)}, 76(3):495--580.

\bibitem[Galeano et~al., 2006]{2006}
Galeano, P., Peña, D., and Tsay, R.~S. (2006).
\newblock Outlier detection in multivariate time series by projection pursuit.
\newblock {\em Journal of the American Statistical Association}, 101(474):pp.
  654--669.

\bibitem[Jandhyala et~al., 2013]{JTSA:JTSA12035}
Jandhyala, V., Fotopoulos, S., MacNeill, I., and Liu, P. (2013).
\newblock Inference for single and multiple change-points in time series.
\newblock {\em Journal of Time Series Analysis}.

\bibitem[Jeng et~al., 2013]{rare_copy_number}
Jeng, X.~J., Cai, T.~T., and Li, H. (2013).
\newblock Simultaneous discovery of rare and common segment variants.
\newblock {\em Biometrika}, 100(1):157--172.

\bibitem[Jin, 2004]{jin}
Jin, J. (2004).
\newblock {\em Detecting a target in very noisy data from multiple looks},
  volume Volume 45 of {\em Lecture Notes--Monograph Series}, pages 255--286.
\newblock Institute of Mathematical Statistics, Beachwood, Ohio, USA.

\bibitem[Kulkarni, 2012]{kulkarni2012introduction}
Kulkarni, V. (2012).
\newblock {\em Introduction to Modeling and Analysis of Stochastic Systems}.
\newblock Springer Texts in Statistics. Springer London, Limited.

\bibitem[Levine and Casella, 2001]{Levine2001}
Levine, R.~A. and Casella, G. (2001).
\newblock {Implementations of the Monte Carlo EM Algorithm}.
\newblock {\em Journal of Computational and Graphical Statistics},
  10(3):422--439.

\bibitem[Olshen et~al., 2004]{Olshen01102004}
Olshen, A.~B., Venkatraman, E.~S., Lucito, R., and Wigler, M. (2004).
\newblock Circular binary segmentation for the analysis of array based {DNA}
  copy number data.
\newblock {\em Biostatistics}, 5(4):557--572.

\bibitem[Pinto et~al., 2011]{Pinto}
Pinto, D., Darvishi, K., Shi, X., Rajan, D., Rigler, D., Fitzgerald, T.,
  Lionel, A.~C., Thiruvahindrapuram, B., MacDonald, J.~R., Mills, R., Prasad,
  A., Noonan, K., Gribble, S., Prigmore, E., Donahoe, P.~K., Smith, R.~S.,
  Park, J.~H., Hurles, M.~E., Carter, N.~P., Lee, C., Scherer, S.~W., and Feuk,
  L. (2011).
\newblock Comprehensive assessment of array-based platforms and calling
  algorithms for detection of copy number variants.
\newblock {\em Nature Biotechnology}, 29(6):512--521.

\bibitem[Qu et~al., 2005]{1387011}
Qu, G., Hariri, S., and Yousif, M. (2005).
\newblock Multivariate statistical analysis for network attacks detection.
\newblock In {\em Computer Systems and Applications, 2005. The 3rd ACS/IEEE
  International Conference on}, pages 9--.

\bibitem[Siegmund et~al., 2011]{siegmund2011}
Siegmund, D., Yakir, B., and Zhang, N.~R. (2011).
\newblock Detecting simultaneous variant intervals in aligned sequences.
\newblock {\em The Annals of Applied Statistics}, 5(2A):645--668.

\bibitem[Spiegel et~al., 2011]{Spiegel:2011:PRC:2003653.2003657}
Spiegel, S., Gaebler, J., Lommatzsch, A., De~Luca, E., and Albayrak, S. (2011).
\newblock Pattern recognition and classification for multivariate time series.
\newblock In {\em Proceedings of the Fifth International Workshop on Knowledge
  Discovery from Sensor Data}, SensorKDD '11, pages 34--42, New York, NY, USA.
  ACM.

\bibitem[Tsay et~al., 2000]{Tsay01122000}
Tsay, R.~S., Peña, D., and Pankratz, A.~E. (2000).
\newblock Outliers in multivariate time series.
\newblock {\em Biometrika}, 87(4):789--804.

\bibitem[Wyse et~al., 2011]{WYSE2011}
Wyse, J., Friel, N., and Rue, H. (2011).
\newblock Approximate simulation-free {B}ayesian inference for multiple
  changepoint models with dependence within segments.
\newblock {\em Bayesian Analysis}, 6(4):501--528.

\bibitem[Yau and Holmes, 2010]{loss_holmes}
Yau, C. and Holmes, C.~C. (2010).
\newblock {A decision theoretic approach for segmental classification using
  Hidden Markov models}.

\bibitem[Zhang, 2010]{overview_CNV}
Zhang, N. (2010).
\newblock {DNA} copy number profiling in normal and  tumor genomes.
\newblock In Feng, J., Fu, W., and Sun, F., editors, {\em Frontiers in
  Computational and Systems Biology}, volume~15 of {\em Computational Biology},
  pages 259--281. Springer London.

\bibitem[Zhang et~al., 2010]{detec_simul_gene_copy}
Zhang, N.~R., Siegmund, D.~O., Ji, H., and Li, J.~Z. (2010).
\newblock Detecting simultaneous changepoints in multiple sequences.
\newblock {\em Biometrika}, 97(3):631--645.

\end{thebibliography}
\end{document}